\documentclass[12pt]{article}

\usepackage[svgnames,dvipsnames]{xcolor}
\usepackage{graphicx}
\usepackage{amsmath,mathtools,amssymb,setspace,amsfonts,enumitem,fullpage,amsthm,comment,color,bm,nth}
\usepackage{tikz}
\usepackage {palatino}
\usepackage[colorlinks,pagebackref=false]{hyperref}
\usepackage[longnamesfirst]{natbib}
\setcitestyle{maxnames=3}

\usepackage[compact,small,raggedright]{titlesec}
\titlelabel{\thetitle.\quad}
\usepackage[font={small,it},labelfont=bf,labelsep=period]{caption}
\usepackage{enumitem}
\setlist{noitemsep,parsep=0pt,partopsep=0pt,topsep=0pt}
\usepackage[nottoc]{tocbibind} 

\let \savenumberline \numberline
\def \numberline#1{\savenumberline{#1.}}

\setlist[enumerate]{itemsep=5pt}

\let\oldfootnote\footnote
\renewcommand\footnote[1]{\oldfootnote{\hspace{.5mm}#1}}
\makeatletter
\renewenvironment{proof}[1][\proofname] {\par\pushQED{\qed}\normalfont\topsep6\p@\@plus6\p@\relax\trivlist\item[\hskip\labelsep\bfseries#1\@addpunct{.}]\ignorespaces}{\popQED\endtrivlist\@endpefalse}
\makeatother

\definecolor{dark-red}{rgb}{0.4,0.15,0.15}
\definecolor{dark-blue}{rgb}{0.15,0.15,0.75}
\definecolor{medium-blue}{rgb}{0,0,0.5}

\hypersetup{
    pdftitle={Multi-Sender Disclosure with Costs},
    pdfauthor={Kartik, Lee, Suen},
    citecolor=DarkBlue,
    bookmarksnumbered=true,
    urlcolor=Indigo,
    linkcolor=DarkBlue,
}

\newtheorem{corollary}{Corollary}
\newtheorem{lemma}{Lemma}

\newtheorem{proposition}{Proposition}

\theoremstyle{definition}

\newtheorem{example}{Example}
\AtEndEnvironment{example}{\hfill$\diamond$}
\newtheorem{claim}{Claim}[example]
\newtheorem{remark}{Remark}
\AtEndEnvironment{remark}{\hfill$\diamond$}

\newenvironment{appendixref}[1]{\hyperref[#1]{Appendix \ref{#1}}}

\DeclareMathOperator*{\argmin}{arg\,min}
\DeclareMathOperator{\sign}{sign}
\DeclarePairedDelimiter{\abs}{\lvert}{\rvert}
\renewcommand{\bar}{\overline}
\providecommand{\abs}[1]{\lvert#1\rvert}

\def \Reals{\mathbb R}

\renewcommand{\epsilon}{\varepsilon}
\def \union{\cup}
\newcommand{\ind}{{1}\hspace{-2.5mm}{1}} 

\newcommand{\case}[1]{\textsc{(#1)}}

\newcommand{\E}{\mathbb{E}}

\newcommand{\mailto}[1]{\href{mailto:#1}{\texttt{#1}}} 
\makeatletter
\newcommand{\citepos}{\@ifstar\citepos@star\citepos@nostar}
\newcommand{\citepos@nostar}[1]{\citeauthor{#1}'s \citeyearpar{#1}}
\newcommand{\citepos@star}[1]{\citeauthor*{#1}'s \citeyearpar{#1}}
\makeatother

\begin{document}

\begin{titlepage}

\title{\textbf{\Large Multi-Sender Disclosure with Costs}\thanks{We thank Nageeb Ali, Ilan Guttman, Keri Hu, Hongcheng Li, Alessandro Lizzeri, Steve Matthews, Ariel Pakes, Andrea Prat, \href{https://www.refine.ink}{Refine.ink}, 
Mike Riordan, Satoru Takahashi, Jidong Zhou, and various seminar and conference audiences for comments on this paper or related work.  Kartik gratefully acknowledges financial support from the NSF (Grant SES-1459877).}}
\author{Navin Kartik\thanks{Yale University, Department of Economics. Email: \href{mailto:nkartik@gmail.com}%
{\texttt{nkartik@gmail.com}}. The author was affiliated with Columbia University during most of the work on this project.}  \and {Frances Xu Lee}\thanks{Quinlan School of Business, Loyola University Chicago. Email: \mailto{francesxu312@gmail.com}.} \and {Wing Suen}\thanks{Faculty of Business and Economics,
The University of Hong Kong. Email: \mailto{wsuen@econ.hku.hk}.}}
\maketitle

\begin{abstract}
\noindent We study voluntary disclosure with multiple biased senders who may bear costs for disclosing or concealing their private information. Under relevant assumptions, disclosures are strategic substitutes under a disclosure cost but complements under a concealment cost. Additional senders thus impede any sender's disclosure under a disclosure cost but promote it under a concealment cost. In the former case, a decision maker can be harmed by additional senders, even when senders have opposing interests. The effects under both kinds of message costs turn on how a sender, when concealing his information, expects others' messages to systematically sway the decision maker's belief.
\end{abstract}

\bigskip

\thispagestyle{empty}
\end{titlepage}

\newpage

\begingroup
\singlespacing
\addtocontents{toc}{\protect\setcounter{tocdepth}{2}}
\addtocontents{toc}{~\hfill\textbf{Page}\par}
\endgroup

\newpage
\onehalfspacing
\setcounter{page}{1}

\section{Introduction}
\label{sec:intro}

Decision makers routinely rely on multiple interested parties for information. Judges hear arguments from opposing counsel; consumers encounter product claims from competing sources; and legislators receive evidence from various interest groups. In each of these settings, information senders strategically choose what to disclose and what to withhold. How are one sender's disclosure incentives affected by the presence of other senders? What are the implications for the decision maker's welfare? Does competition among senders promote the revelation of information---and does the answer depend on whether their interests are aligned or opposed?

We study these questions in a simple \emph{voluntary disclosure} game. Biased senders are endowed with certifiable private information---evidence---and choose whether to reveal it. That is, senders cannot lie but can conceal. Most prior work on multi-sender disclosure assumes that informed senders have identical information.\footnote{See, among others, \citet{MR86}, \citet{LS95}, \citet{Shin98}, \citet{BJ11}, and \citet{BM13}. Two exceptions are \citet{OFPS90} and \citet{HKPR13}, who identify conditions for full disclosure and address different questions than we do. } We instead assume imperfectly correlated information: specifically, conditional on an underlying state that is payoff-relevant to a decision maker (DM), senders draw independent signals. As in \citet{Dye85}, senders may also be uninformed; this prevents ``unraveling'' \citep{GH80,Milgrom81}. We allow senders to have either opposing or similar biases, but assume each sender's payoff is state-independent and linear (increasing or decreasing) in the DM's belief.

Our focus is on how strategic disclosure interacts with different types of costs. On the one hand, preparing and disseminating information may be costly. A firm facing litigation must hire attorneys to sort through documents, identify relevant materials, and redact appropriately. Certain documents require notarization to prevent fraud. There are also ``proprietary costs'' \citep{Verrecchia83}: disclosing a beverage's ingredients may reveal a secret recipe to competitors; disclosing a firm's financial information to investors may simultaneously inform regulators and rivals. We call costs directly tied to the act of furnishing evidence \emph{disclosure costs}. On the other hand, suppression of information is also sometimes costly. Beyond any resource costs, there may be a psychic disutility from concealing information, or concealment may be discovered ex post---by auditors, whistleblowers, or happenstance---and result in punishment or reputational damage. For instance, in January 2015, the National Highway Traffic Safety Administration fined Honda \$70 million because ``it did not report hundreds of death and injury claims \ldots for the last 11 years'' \citep{NYT-Honda}. Relatedly, there is a duty to disclose in various business, financial, and legal contexts. We call costs tied to not disclosing evidence \emph{concealment costs}.\footnote{Studies of mandatory disclosure \citep[e.g.,][]{MP85,Shavell94,DGP09} can be viewed as dealing with infinitely large concealment costs.} 
Naturally, both disclosure and concealment costs can coexist; what is relevant for our paper is the \emph{net} message cost, which we classify as a disclosure or concealment cost depending on which dominates.

In the presence of a message cost, each sender must consider how others' information disclosure will affect the DM's posterior belief, and how that varies with his own message.
We leverage our model's structure and an insight from \citet{KLS21-IVP} on how agents with heterogeneous priors expect additional information to affect each other's beliefs. We obtain a unified treatment of multi-sender disclosure---regardless of whether senders have similar or opposing biases, and regardless of the nature of the message cost. In turn, that yields new insights into the classic question of when competition promotes disclosure.

\paragraph{Results.} Our paper's central logic is as follows. Any sender's equilibrium disclosure behavior follows a threshold rule of disclosing all sufficiently favorable signals. We first establish a simple but important benchmark: without a message cost, any sender uses the same disclosure threshold as he would in a (hypothetical) single-sender game. In other words, there is a \emph{strategic irrelevance}. The intuition is that without a message cost, a sender's objective is the same regardless of the presence of other senders: he simply wants to induce the most favorable ``interim belief'' in the DM based on his own message, as this will lead to the most favorable posterior belief based on all senders' messages.

How do message costs alter the irrelevance result? Consider a concealment cost. In a single-sender setting, a sender $i$'s disclosure threshold will be such that the DM's interpretation of nondisclosure is more favorable than $i$'s private belief at the threshold---a wedge necessary to compensate $i$ for the concealment cost. Nondisclosure thus generates an interim \emph{disagreement} between the DM's belief and the threshold type's private belief. Naturally, disclosure produces no such disagreement. Now add a second sender $j$ to the picture. \citepos*{KLS21-IVP} theorem on \emph{information validates the prior} (IVP) provides a key insight: the threshold type of $i$ predicts that $j$'s message will, on average, make the DM's posterior less favorable than the DM's interim belief following $i$'s nondisclosure. Intuitively, sender $i$ expects the other's communication to move the DM's belief towards $i$'s estimate of the truth, which is $i$'s private belief. Consequently, concealment is now less attractive to sender $i$. IVP further implies that $j$'s effect on $i$ is stronger when $j$ discloses more, i.e., when $j$ is more informative. In sum, senders' disclosures are \emph{strategic complements} under a concealment cost.\footnote{In a different model, \citet{BJ11} find an effect related to that we find under concealment cost. Loosely speaking, ``reputation loss'' in their model plays a similar role to concealment cost in ours.}

The logic reverses under a disclosure cost. In a single-sender setting, the DM's interpretation of nondisclosure is now less favorable than $i$'s private belief at the threshold---the gain from disclosing information must compensate $i$ for the direct cost. Reasoning analogously to above, IVP now implies that disclosure becomes less attractive when the other sender is more informative: the threshold type of $i$ expects the other sender's message to make the DM's posterior more favorable to $i$, reducing the gains from disclosure. Consequently, senders' disclosures are \emph{strategic substitutes} under a disclosure cost.

\paragraph{Implications.} When there are costs of concealment, a DM always benefits from an additional sender not only because of the information this sender provides, but also indirectly because it improves disclosure from other senders. With disclosure costs, however, the strategic substitution result implies that while a DM gains some direct benefit from consulting an additional sender, the indirect effect on other senders' behavior is detrimental to the DM. In general, the net effect is ambiguous; it is not hard to construct examples in which the DM is made strictly worse off by adding a sender, even if this sender has an opposing bias to that of an existing sender. Thus, competition between senders need not increase information revelation nor benefit the DM.\footnote{A number of prior papers offer formal analyses supporting the viewpoint that competition between senders helps a DM; a sample of the varied settings includes voluntary disclosure with common information \citep{MR86}, cheap talk \citep{KM01,Battaglini02}, Bayesian persuasion \citep{GK12-Competition,AK20}, and information acquisition \citep{DT99}. 
\citet{CDI12} present a result in which increased competition leads to less voluntary disclosure. Their model can be viewed as one in which senders bear a concealment cost that is assumed to decrease in the amount of disclosure by other senders. \citet{EGK12} show how a DM can be harmed by ``information improvements'' in a cheap-talk setting, but the essence of their mechanism is not the strategic interaction between senders.} The DM can even be made worse off when the disclosure cost becomes lower or a sender is more likely to be informed, although either change would help the DM when facing just one sender.

We interpret the perverse welfare results under disclosure costs as cautionary for applications. Institutional changes that appear beneficial at first blush may be detrimental. For example, given the importance of disclosure costs in arbitration and litigation \citep[e.g.,][]{Sobel89}, our results qualify arguments made in favor of adversarial procedures based on promoting information revelation \citep[e.g.,][]{Shin98}.

It bears emphasis that our results are driven by two forces: (i) message costs create equilibrium disagreement upon nondisclosure between a sender and the DM; and (ii) when others disclose more, this disagreement and the IVP force combine to produce a systematic shift in a sender's own disclosure behavior. The direction of this shift depends only on the direction of disagreement, which in turn is determined by the nature of message costs. For this reason, it does not matter whether a sender seeks to push the DM's posterior up or down. That also clarifies that our central logic is not one of free-riding or pivotality: we can have strategic substitution even when senders have opposing biases, which runs counter to one sender free-riding on the other's ``contribution'' to the same goal; yet we can also have strategic complementarity, which runs counter to reduced pivotality.

\paragraph{Other related literature.} \citet{Milgrom08} and \citet{DJ10} survey the disclosure literature. Our reduced-form approach to message costs is in line with how disclosure costs are typically modeled, dating back to \citet{Jovanovic82} and \citet{Verrecchia83}. \citet{EF19} study two senders with opposing biases and disclosure costs; unlike us, they assume perfectly correlated signals, which leads to strategic interaction effects even without a message cost (cf.~\citet{BM13} and \appendixref{app:correlated}).

We are not aware of prior general treatments of concealment costs. However, some recent work models a cost of concealment in specific ways motivated by their applications.  \citet{MV16} analyze litigation risk in a dynamic setting, showing that the threat of penalties for nondisclosure can lead a manager to voluntarily disclosing bad news to preempt exogenous news arrival. \citet{Dye17} also studies a single sender, but in a static model where concealment may be detected, triggering damages proportional to a buyer's overpayment.  \citet{DR18} study teams of prosecutors facing Brady obligations, incorporating both formal sanctions and moral costs. In all these papers, the penalties for concealment depend on equilibrium beliefs or endogenous actions. Our reduced-form approach sacrifices institutional detail but gains generality: it allows us to nest both concealment and disclosure costs within a unified framework and isolate the (net) cost structure as the key determinant of how one sender's information disclosure affects others' behavior.

As already noted, our analysis relies on \citepos*{KLS21-IVP} IVP theorem. \citet{AD25} experimentally test that result. The current paper leverages the IVP effect to study how disclosure decisions are affected by message costs. \citet*{KLS21-IVP} present applications to testing and signaling, while \citet{KLS17-Acquisition} study implications for {information acquisition}.

\paragraph{Outline.} \autoref{sec:model} presents the model, with only two senders for simplicity. \autoref{sec:single} develops the single-sender benchmark. \autoref{sec:main} establishes our main results on strategic substitutes and complements. \autoref{sec:extensions} discusses extensions, including to many senders. \autoref{sec:conclusion} concludes. Proofs and some additional discussion are in the \hyperref[sec:appendix]{Appendix}; there is also a \hyperref[app:nonlinear-supp]{Supplementary Appendix} with further material.

\section{Model}
\label{sec:model}

\paragraph{Players.}
There is an unknown state of the world, $\omega\in \{0,1\}$. A decision maker, DM hereafter, will form a belief $\beta_{DM}$ that $\omega=1$. (Throughout, all beliefs refer to the probability of state $1$.) For much of our analysis, all that matters is the belief that the DM holds. For welfare evaluation, however, it is useful to view the DM as taking an action $a$ with von-Neumann--Morgenstern utility function $u_{DM}(a,\omega)$. There are two senders, indexed by $i\in \{1,2\}$. (\autoref{sec:manysenders} generalizes to many senders.) In a reduced form, each sender $i$ has state-independent preferences over the DM's belief given by the von Neumann--Morgenstern utility function $u(\beta_{DM},b_i)=b_i \beta_{DM}$, where $b_i \in \{-1, 1\}$ captures a sender's bias. That is, each sender has linear preferences over the DM's expectation of the state; $b_i = 1$ means that sender $i$ is biased upward, and conversely for $b_i = -1$. Senders' biases are common knowledge. We say that two senders have \emph{similar biases} if their biases have the same sign, and \emph{opposing biases} otherwise.

\paragraph{Information.}
The DM relies on the senders for information. All players share a common prior $\pi\in (0,1)$ over the state. Each sender may receive some private information about the state. Specifically, following \citet{Dye85}, with independent probability $p_i \in (0,1)$, a sender $i$ is informed and receives a signal $s_i\in S$; with probability $1-p_i$, he is uninformed, in which case we denote $s_i=\phi$. If informed, sender $i$'s signal is drawn from a distribution that depends upon the true state, but independently of the other sender's signal conditional on the state. Without loss, we equate an informed sender's signal with his \emph{private belief}, i.e., a sender's posterior on state $\omega=1$ given only his own signal $s\neq \phi$ (as derived by Bayesian updating) is $s$. For convenience, we assume the cumulative distributions of an informed sender's signals in each state, $F(s|\omega)$ for $\omega\in \{0,1\}$, have common support $S:=[\underline s,\overline s]\subseteq [0,1]$ and admit respective densities $f(s|\omega)$. It would be straightforward to allow $F(\cdot)$ to vary across senders, but we abstract from such heterogeneity to reduce notation.

\paragraph{Communication.}
Signals are ``hard evidence''; a sender with signal $s_i\in S\union \{\phi\}$ can send a message $m_i\in \{s_i,\phi\}$. In other words, an uninformed sender only has one message available, $\phi$, while an informed sender can either report his true signal or feign ignorance by sending the message $\phi$.\footnote{Due to the senders' monotonic preferences, standard ``skeptical posture'' arguments imply that our results would be unaffected if we were to allow for a richer message space, for example if an informed sender could report any subset of the signal space that contains his true signal. Likewise, allowing for cheap talk would not affect our results as cheap talk cannot be influential in equilibrium.} We refer to any message $m_i\neq \phi$ as \emph{disclosure} and the message $m_i=\phi$ as \emph{nondisclosure}. When an informed sender chooses nondisclosure, we say he is \emph{concealing} information. That senders must either tell the truth or conceal their information is standard; a justification is that signals are verifiable and large penalties will be imposed on a sender if a reported signal is discovered to be untrue. Note that being uninformed is not verifiable.

\paragraph{Message costs.}
A sender $i$ who sends message $m_i \neq \phi$ bears a known utility cost $c\in \mathbb{R}$. We take this cost to be common across senders for notational simplicity.
We refer to $c>0$ as a (net) \emph{disclosure cost} and $c<0$ as a (net) \emph{concealment cost}; see the introduction for interpretation and discussion. As is well known, a disclosure cost precludes full disclosure \citep{Jovanovic82,Verrecchia83}. For this reason, our conclusions under $c>0$ do not require the assumption that a sender may be uninformed (i.e., we could allow $p_i=1$ in that case). We maintain that assumption to provide a unified treatment of both disclosure cost ($c>0$) and no cost or concealment cost ($c\leq 0$). In the latter cases, there would be full disclosure (``unraveling'') were $p_i=1$.

\paragraph{Timing.}
The game unfolds as follows: nature initially determines the state $\omega$ and then conditionally (on the realized state) independently draws each sender $i$'s private information, $s_i\in S \union \{\phi\}$; senders then simultaneously send their respective messages $m_i$ to the DM (whether messages are public or privately observed by the DM is irrelevant); the DM then forms her belief, $\beta_{DM}$, according to Bayes' rule, whereafter each sender $i$'s payoff is realized as
\begin{equation}
\label{e:payoff}	
b_i \beta_{DM}-c\cdot\ind{\{m_i\neq \phi\}}.
\end{equation}
All aspects of the game except the state and senders' signals (or lack thereof) are common knowledge. Our solution concept is the natural adaptation of perfect Bayesian equilibrium, which we will refer to simply as ``equilibrium.''\footnote{That is: (1) the receiver forms her belief using Bayes rule on path (note that the message $m_i=\phi$ is necessarily on path), treating any off-path message $m_i \in S$ as proving signal $s_i=m_i$; and (2) each sender chooses his message optimally given his information, the other senders' strategies, and the receiver's belief updating.} The notion of welfare for any player is ex-ante expected utility.

\section{A Single-Sender Benchmark}
\label{sec:single}

As a preliminary step and benchmark, begin by considering a (hypothetical) game between a single sender $i$ and the DM. Our analysis in this subsection generalizes some existing results in the literature. For concreteness, suppose the sender is upward biased; straightforward analogs of the discussion below apply if the sender is downward biased.

For any $\beta\in [0,1]$, define \mbox{$f_\beta (s):=\beta f(s|1)+(1-\beta)f(s|0)$} as the unconditional density of signal $s$ given a belief that puts probability $\beta$ on state $\omega=1$. Let $F_\beta$ be the corresponding cumulative distribution. Since the sender has private belief $s$ upon receiving signal $s$, disclosure of signal $s$ will lead to the DM also holding belief $s$. It follows that given any \emph{nondisclosure belief}, i.e., the DM's posterior belief when there is nondisclosure, the optimal strategy for the sender---if informed, as otherwise he can only send message $\phi$---is a threshold strategy of disclosing all signals above some \emph{disclosure threshold}, say $\hat s$, and concealing all signals below it. 
Suppose the sender uses a disclosure threshold $\hat s$. Define the function $\eta:[0,1]\times (0,1)\times [0,1]\rightarrow [0,1]$ by
\newcommand{\ubstrut}{\vphantom{\dfrac{pF_{\pi}(\hat s)}{1-p+pF_{\pi}(\hat s)}}}
\newcommand{\grayub}[2]{%
  {\color{Gray}\underbrace{\color{black}{#1}\ubstrut}_{\text{#2}}}%
}
\begin{align}
\eta(\hat s,p,\pi)
:=\;
\grayub{\left(\frac{1-p}{1-p +p F_{\pi}(\hat s)}\right)}{\shortstack{Posterior prob\\that uninformed}}
\;
\times \grayub{\pi}{Prior}
\;+\;
\grayub{\left(\frac{p F_{\pi}(\hat s)}{1-p +p F_{\pi}(\hat s)}\right)}{\shortstack{Posterior prob\\that concealed}}
\; \times \;
\grayub{\E_{\pi}\!\left[s \mid s<\hat s\right]}{\shortstack{Average \\ concealed signal}},
\label{e:NDbelief}
\end{align}
where $\E_\pi[\cdot]$ denotes expectation with respect to $F_\pi$. The function $\eta(\hat s,p,\pi)$ is the DM's posterior following nondisclosure, derived from Bayes' rule given conjectured threshold $\hat s$, probability $p$ that the sender is informed, and prior $\pi$.

An increase in the sender's disclosure threshold has two effects on the nondisclosure belief. First, it raises the likelihood that nondisclosure stems from concealment rather than the sender being uninformed. Second, conditional on concealment, it raises the expected signal. As the DM's belief conditional on concealment is below the prior (because the sender is using threshold strategy), the two effects work in opposite directions; the second effect dominates if and only if $\hat{s} > \eta(\hat{s},p,\pi)$. On the other hand, holding the disclosure threshold fixed, an increase in the probability of the sender being informed has an unambiguous effect, as it raises the likelihood that nondisclosure reflects concealment.

\begin{lemma}
\label{lem:NDbelief}
The nondisclosure belief function $\eta(\hat s,p,\pi)$ has the following properties:
\begin{enumerate}
\item \label{NDbelief1} It is strictly decreasing in $\hat s$ when $\hat s<\eta(\hat s,p,\pi)$ and strictly increasing when $\hat s>\eta(\hat s,p,\pi)$. Consequently, the unique solution to $\hat s = \eta(\hat s,p,\pi)$ is $\argmin_{\hat s} \eta(\hat s,p,\pi)\in (\underline s,\pi)$.
\item \label{NDbelief2} It is weakly decreasing in $p$, strictly if $\hat s\in(\underline s,\overline s)$.
\end{enumerate}
\end{lemma}

See \autoref{fig:single}, which depicts both parts of the lemma and the subsequent results in this section.  The quasiconvexity in part \ref{NDbelief1} of \autoref{lem:NDbelief} will be crucial for our comparative statics. The lemma's other properties have appeared previously in \citet[Proposition 1 and subsequent discussion]{ADK11}; see also \citet[Proposition 1]{DKS19}.

\begin{figure}
\centering
\begin{tikzpicture}[scale=0.9, transform shape, x=10cm, y=8.35cm]

\def\axiswidth{1.2pt}
\def\curvewidth{1.1pt}
\def\thinwidth{0.4pt}

\draw[line width=\axiswidth] (0,0) rectangle (1,1);

\draw[line width=\axiswidth,->] (0,1) -- (0,1.08);
\draw[line width=\axiswidth,->] (1,0) -- (1.08,0);

\node[above left] at (0,0.96) {$\bar s$};
\node[below left] at (0,0.035) {$\underline s$};
\node[below right] at (0.967,0.005) {$\bar s$};
\node[below] at (1.11,0.038) {$s_i$};


\draw[line width=\thinwidth] (0.0,0.0) -- (1,1);
\node at (0.13,0.17) {$s_i$};

\draw[dash pattern=on 1.5pt off 1.5pt,line width=\thinwidth] (0.08,0.0) -- (1,0.92);
\node at (0.3,0.15) {$s_i-c$};

\draw[dash dot,line width=\thinwidth] (0.0,0.12) -- (0.88,1.0);
\node at (0.615,0.83) {$s_i-c'$};


\draw[line width=\curvewidth]
  plot[smooth]
  coordinates {
    (0.00,0.80)
    (0.09,0.56)
    (0.145,0.45)
    (0.163,0.4192)
    (0.215,0.36)
    (0.33,0.33)
    (0.43,0.347)
    (0.54,0.4132)
    (0.63,0.55)
    (0.73,0.683)
    (0.835,0.7545)
    (1,0.8)
  };
\node at (0.145,0.7) {$\eta(s_i,p_i,\pi)$};

\draw[line width=\curvewidth]
  plot[smooth
  ]
  coordinates {
    (0.00,0.80)
    (0.024,0.65)
    (0.073,0.50)
    (0.15,0.32)
    (0.28,0.28)
    (0.47,0.30)
    (0.65,0.42)
    (0.80,0.55)
    (0.93,0.70)
    (1.00,0.80)
  };
\node at (0.67,0.33) {$\eta(s_i,p_i',\pi)$};

\node[left] at (0.0,0.8) {$\pi$};
\end{tikzpicture}
\caption{The single-sender game with an upward-biased sender, illustrated with $c>0>c'$ and $p'_i>p_i$. Equilibrium thresholds are given by intersections of $\eta(s_i,\cdot)$ and $s_i-c$ or $s_i-c'$.}\label{fig:single}
\end{figure}
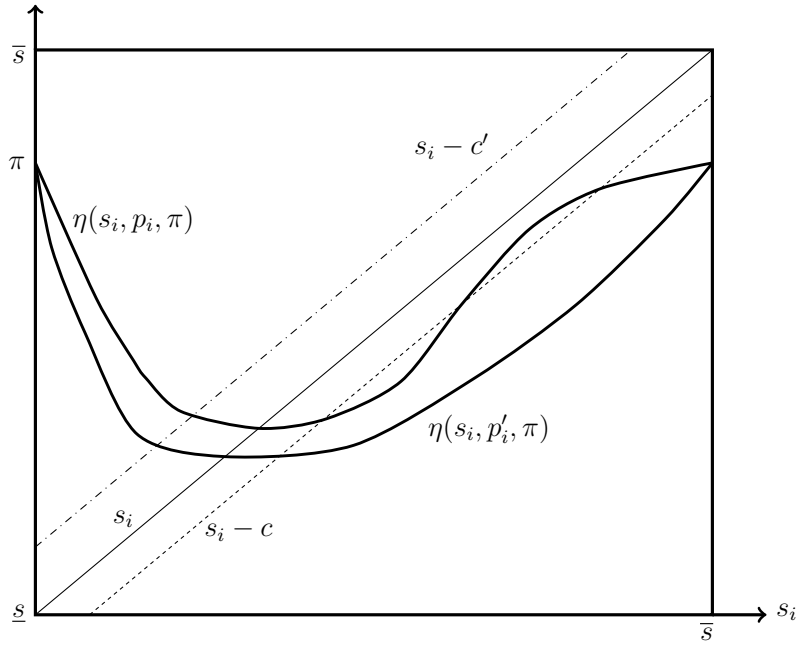

It follows that any equilibrium is fully characterized by the sender's disclosure threshold.  If this threshold is interior, the sender must be indifferent between disclosing the threshold signal and concealing it.  As sender $i$'s payoff from disclosing signal $s_i$ is $ s_i-c$, we obtain the following characterization.

\begin{proposition}
\label{prop:single}
Assume there is one sender $i$ who is biased upward.
\begin{enumerate}
\item \label{single1} Any equilibrium has a disclosure threshold $\hat s_i^0$ satisfying one of: (i) $\hat s_i^0$ is interior and $\eta(\hat s_i^0,p_i,\pi)= \hat s_i^0-c$; (ii) $\hat s_i^0=\underline s$ and $\pi\leq \underline s - c$; or (iii) $\hat s_i^0=\overline s$ and $\pi\geq \overline s-c$. Conversely, any $\hat s_i^0$ satisfying (i), (ii), or (iii) is an equilibrium threshold.
\item \label{single2} If there is no message cost or a concealment cost ($c\leq 0$), equilibrium is unique. Moreover, the equilibrium threshold is interior if there is no message cost ($c=0$).
\item \label{single3} If there is a disclosure cost ($c>0$), there can be multiple equilibria.
\end{enumerate}
\end{proposition}

Part \ref{single1} of \autoref{prop:single} is straightforward; parts \ref{single2} and \ref{single3} build on \autoref{lem:NDbelief}.  Multiple equilibria can arise under a disclosure cost because, in the relevant domain (to the right of the fixed point of $\eta(\cdot,p_i,\pi)$), the DM's nondisclosure belief is increasing in the sender's disclosure threshold.  When there are multiple equilibria, we will focus on the highest and lowest equilibria in terms of the disclosure threshold.  As a higher-threshold equilibrium is less (Blackwell) informative, these extremal equilibria are, respectively, worst and best for the DM's welfare. The ranking of these equilibria is reversed for the sender's welfare, by linearity of his preferences. To elaborate: since in expectation the DM's belief is $\pi$ in any equilibrium, the sender's welfare in an equilibrium with threshold $\hat s_i^0$ is $ \pi-p_i (1-F_{\pi}(\hat s_i^0))c$. Consequently, given that multiplicity requires $c>0$, the sender ex-ante prefers a higher disclosure threshold, as it implies a lower probability of incurring the disclosure cost.

A key observation concerns \emph{disagreement}: the divergence between the sender's threshold belief $\hat s^0_i$ and the DM's nondisclosure belief $\eta(\hat s^0_i,\cdot)$. When $c=0$, these coincide. When $c\neq 0$, they differ in any equilibrium. Specifically, if $c>0$, then $\hat s_i^0 > \eta(\hat s_i^0,\cdot)$, i.e., the sender's threshold belief exceeds the DM's nondisclosure belief. If $c<0$, the opposite holds. This disagreement will prove crucial. Note that, consistent with \citet{Aumann76}, the disagreement is not common knowledge---the sender knows the DM's belief but not vice versa.

\autoref{prop:single} is stated for an upward-biased sender. For a downward-biased sender, the analysis is symmetric: he discloses all signals below some threshold. The nondisclosure belief becomes
\begin{align*}
\frac{1-p}{1-p +p (1-F_{\pi}(\hat s))}\pi + \frac{p (1-F_{\pi}(\hat s))}{1-p +p (1-F_{\pi}(\hat s))} \E_{\pi}
\left[s \mid s>\hat s\right],
\end{align*}
which is strictly quasiconcave in $\hat{s}$. The equilibrium condition becomes that this expression equals $\hat s+c$, since the sender's payoff is $-s-c$ if he discloses and minus the nondisclosure belief if he conceals. As before, equilibrium is unique when $c \le 0$, while there can be multiple equilibria when $c > 0$. For $c\neq 0$, the direction of disagreement is now reversed: with a downward-biased sender, $c>0$ implies the sender's threshold belief is {lower} than the DM's nondisclosure belief, and conversely for $c<0$. Furthermore, a lower threshold now corresponds to \emph{less} disclosure.

The following comparative statics hold with an upward-biased sender; the modifications for a downward-biased sender are straightforward in light of the above discussion.

\begin{proposition}
\label{prop:singleCS}
Assume there is one sender who is upward biased.
\begin{enumerate}
\item A higher probability of being informed leads to more disclosure: the highest and lowest equilibrium thresholds (weakly) decrease.
\item A higher disclosure cost or a lower concealment cost (in magnitude) leads to less disclosure: the highest and lowest equilibrium thresholds (weakly) increase.
\end{enumerate}
\end{proposition}

The logic for the first part follows from \autoref{lem:NDbelief}: given any conjectured threshold, a higher $p_i$ leads to a lower nondisclosure belief, increasing the sender's gain from disclosure. For $c=0$, this comparative static appears in \citet{JK88} and \citet*{ADK11}. The second part of \autoref{prop:singleCS} is straightforward, as a larger message cost $c$ makes disclosure less attractive. Since scaling the magnitude of $c$ is equivalent to scaling the agent's bias parameter $b_i$ (cf.~expression \eqref{e:payoff}), an equivalent interpretation is that an agent with a stronger persuasion motive discloses more when $c>0$ but less when $c<0$. Both parts of \autoref{prop:singleCS} are seen in \autoref{fig:single}.\footnote{\label{fn:fig-ND}In the figure, $\eta(\cdot)$ has slope less than one at its largest crossing point with $s_i-c$. This makes transparent that an increase in $p_i$ reduces the highest equilibrium threshold. If the slope exceeded one at the highest crossing point, the largest equilibrium threshold would be $\overline s$, and a small increase in $p_i$ would not alter it. We also note that the $\eta(\cdot)$ depicted is valid: it can be shown that for any continuously differentiable $\psi:[\underline s,\overline s]\to [0,1]$ with $\psi(\underline s)=\psi(\overline s) \in (0,1)$ and $\sign[s-\psi(s)]=\sign[\psi'(s)]$ (where $\psi'$ denotes the derivative), there are model parameters---namely, $\pi\in (0,1)$, $p_i\in (0,1)$, $f(s|0)$, and $f(s|1)$---such that $\eta(\cdot,p,\pi)=\psi(\cdot)$. See the proof of \autoref{prop:single} (part \ref{single3}) for details.}

Although we postpone a formal argument to \autoref{sec:main}, it is worth noting now that the foregoing comparative statics have direct welfare implications.  Since the DM prefers more disclosure, a lower message cost and/or a higher probability of the sender being informed (weakly) increases the DM's welfare in a single-sender setting, subject to an appropriate comparison of equilibria (in particular, focusing on the extremal equilibria).

\section{Strategic Substitutes and Complements}
\label{sec:main}

We are now ready to study the two-sender disclosure game. For concreteness, we will suppose that both senders are upward biased; the modifications needed when one or both senders are downward biased are straightforward.

\begin{lemma}
\label{lem:threshold}
Any equilibrium is a threshold equilibrium, i.e., both senders use threshold strategies.
\end{lemma}

In light of \autoref{lem:threshold}, we focus on threshold strategies. Conditionally independent signals mean that the DM's belief updating is separable in the senders' messages. In other words, we can treat it as though the DM first updates from either sender $i$'s message just as in a single-sender model, and then uses this updated belief as an interim prior to update again from the other sender $j$'s message without any further attention to $i$'s message. Thus, given any conjectured pair of disclosure thresholds, $(\hat s_1,\hat s_2)$, there are three relevant nondisclosure beliefs for the DM: if only one sender $i$ discloses his signal $s_i$ while sender $j$ sends message $\phi$, the DM's belief is $\eta(\hat s_j,p_j,s_i)$; if there is nondisclosure from both senders, the DM's belief is $\eta(\hat s_j,p_j,\eta(\hat s_i,p_i,\pi))$.

As discussed earlier, if $i$ discloses his signal then his expectation of the DM's belief---viewed as a random variable that depends on $j$'s message---is $s_i$, no matter what strategy $j$ is using.\footnote{The distribution of the DM's beliefs as a function of $j$'s message depends both on $j$'s strategy and the DM's conjecture about $j$'s strategy.  As the two must coincide in equilibrium, we bundle them to ease exposition.}  On the other hand, if $i$ conceals his signal, then he views the DM as updating from $j$'s message based on the interim belief $\eta(\hat s_i,p_i,\pi)$ that may be different from $s_i$. Denoting an arbitrary message from $j$ by $m_j\in S\union\{\phi\}$, let $\beta(m_j,\hat{s}_j,p_j,s_i)$ denote $i$'s posterior belief about the state given $m_j$ (with the threshold $\hat s_j$ and probability of being informed $p_j$ required to interpret $m_j=\phi$) when $i$'s own interim belief/signal is $s_i$. The DM's posterior after message $m_j$ is a transformation of $i$'s own posterior, where the transformation function depends only on the DM's and $i$'s interim beliefs. Specifically, the DM's posterior after message $m_j$ equals 
$T(\beta(m_j,\hat s_j,p_j,s_i),\eta(\hat s_i,p_i,\pi),s_i)$, 
where for any posterior $\beta_i$ and interim beliefs $\mu_{DM}$ and $\mu_i$, the mapping $T$ is defined by\footnote{\label{fn:transformation}To obtain this transformation, consider an arbitrary (Blackwell) experiment and some signal whose likelihood ratio between states $1$ and $0$ is $l\in (0,\infty)$. 
For any prior $b_0$, Bayes' rule implies that the posterior belief $b(l,b_0)$ satisfies
\begin{equation*}
\frac{b(l,b_0)}{1-b(l,b_0)}
=
\frac{b_0}{1-b_0}\,l .
\end{equation*}
Eliminating the likelihood ratio $l$ yields the following relationship for any two priors $b_0$ and $b_0'$:
\begin{equation*}
b(l,b_0')
=
\frac{b(l,b_0)\,\frac{b_0'}{b_0}}
{b(l,b_0)\,\frac{b_0'}{b_0}
+
\bigl(1-b(l,b_0)\bigr)\,\frac{1-b_0'}{1-b_0}} .
\end{equation*}
This observation that with heterogeneous priors, one agent's posterior can be written as a function of the other agent's posterior and the two priors is also used by \citet{GK14}, \cite{AC16}, and \citet{Galperti19} to different ends.
} 
\begin{equation}
\label{e:transform}
T\left(\beta_i, \mu_{DM}, \mu_i\right):=\frac{\beta_i \frac{\mu_{DM}}{\mu_i}}{\beta_i \frac{\mu_{DM}}{\mu_i}+\left(1-\beta_i\right) \frac{1-\mu_{DM}}{1-\mu_i}}.
\end{equation}
It follows that sender $i$'s expected payoff---his expectation of the DM's posterior belief---should he conceal his signal is
\begin{equation}
\label{e:U}
U(s_i,\hat{s}_i,p_i,\hat{s}_j,p_j)
:=\E_{\hat s_j,p_j}\big[T(
\beta(m_j,\hat{s}_j,p_j,s_i),
\eta(\hat s_i,p_i,\pi),s_i) \mid s_i\big],
\end{equation}
where $\E_{\hat s_j,p_j}$ denotes that the expectation is taken over $m_j$ using the distribution of beliefs that $\hat s_j$ and $p_j$ jointly induce in $i$ about $m_j$ (given $s_i$). As it will be particularly relevant to evaluate $U(s_i,\hat s_i,\cdot)$ when $s_i=\hat s_i$, we will abuse notation and write $U(s_i,p_i,\hat s_j,p_j)$ as shorthand for $U(s_i, s_i,p_i,\hat s_j,p_j)$.

We can now study the ``best response'' of sender $i$ to any disclosure strategy of sender $j$. More precisely, let ${\hat s_i}^{BR}(\hat s_j, p_i, p_j)$ represent the equilibrium disclosure threshold in a (hypothetical) game between sender $i$ and the DM when sender $j$ is conjectured to mechanically adopt disclosure threshold $\hat s_j$; we call this sender $i$'s best response. The threshold $\hat s_i$ is a best response if and only if
\begin{align}
\label{e:BR}
\begin{cases}
 U(\hat s_i,p_i,\hat s_j,p_j) = \hat s_i -c &\text{and } \hat s_i \in (\underline{s},\overline{s}); \text{ or }\\
  U(\underline{s},p_i,\hat s_j,p_j) \le \underline{s}-c& \text{and }\hat s_i = \underline{s}; \text{ or } \\
  U(\overline{s},p_i,\hat s_j,p_j) \ge \overline{s}-c& \text{and } \hat s_i = \overline{s}.  
\end{cases}
\end{align}
The necessity of condition \eqref{e:BR}  is clear; sufficiency follows from the argument in the proof of \autoref{lem:threshold}. In any equilibrium of the overall game, $(s^*_1,s^*_2)$, condition \eqref{e:BR} must hold for each sender $i$ with $\hat s_i=s^*_i$ when his opponent uses $\hat s_j=s^*_j$.

The following key observation underlies all our comparative statics.

\begin{lemma}
\label{lem:rotation}
Consider any $\hat s_j$ and $p_j$.
\begin{enumerate}
\item \label{rotation1} \case{rotation} \ The following implications hold:
\vspace{-.1in}
\begin{align*}
s_i=\eta(s_i,p_i,\pi) & \implies U(s_i,p_i,\hat s_j,p_j) = s_i, \\
s_i > \eta(s_i,p_i,\pi) & \implies \eta(s_i,p_i,\pi) \leq U(s_i,p_i,\hat s_j,p_j) < s_i,\\
s_i < \eta(s_i,p_i,\pi) & \implies s_i < U(s_i,p_i,\hat s_j,p_j) \leq \eta(s_i,p_i,\pi).
\end{align*}
Both weak inequalities in the consequents are strict if and only if $\hat s_j < \overline s$.
\item \label{rotation2} \case{monotonicity} \ If $p'_j \geq p_j$ and $\hat s'_j \leq \hat s_j$, then 
$\abs{U(s_i,p_i,\hat s'_j,p'_j) - s_i} \leq \abs{U(s_i,p_i,\hat s_j,p_j) - s_i}$, with equality if and only if 
$\hat s'_j = \hat s_j$ and either $p'_j = p_j$ or $\hat s_j = \overline{s}$.
\end{enumerate}
\end{lemma}

Both parts of \autoref{lem:rotation} are consequences of the IVP theorem in \citet*[Theorem 1]{KLS21-IVP}. Part \ref{rotation1} reflects $i$'s prediction that, on average, $j$'s message will move the DM's posterior belief away from the DM's interim belief $\eta(\cdot)$ towards $i$'s interim belief $s_i$. Part \ref{rotation2} reflects that there is a larger such shift when $j$'s message is (Blackwell-)more informative, which---as elaborated in the lemma's proof---is the case when $j$ uses a lower threshold (hence discloses more) or has a higher probability of being informed. For some intuition for these effects, consider two extremes: if $j$'s message is completely uninformative (i.e., $j$ never discloses his signal, or $\hat s_j=\bar s$), then the DM's posterior necessarily stays at her interim belief, hence $U(\cdot)=\eta(\cdot)$; whereas if $j$'s message were to reveal the state, then from $i$'s vantage the DM's posterior would equal $1$ with probability $s_i$ and $0$ with probability $1-s_i$, hence $U(\cdot)=s_i$.\footnote{Note that since $j$ does not observe the state, $j$'s message cannot actually reveal the state; this explains the strict inequalities in the consequents of part \ref{rotation1} of \autoref{lem:rotation}.}

\autoref{lem:rotation}'s effects are depicted graphically in \autoref{fig:two}. Part \ref{rotation1} corresponds to comparing the red (short dashed) curve depicting $U(\cdot)$ with the black (solid) curve depicting $\eta(\cdot)$; the former is a rotation of the latter around its fixed point toward the diagonal. Part \ref{rotation2} implies that this rotation is steeper when $p_j$ increases and $\hat s_j$ decreases; this is depicted as the shift from the red curve to the blue (long dashed) curve. 

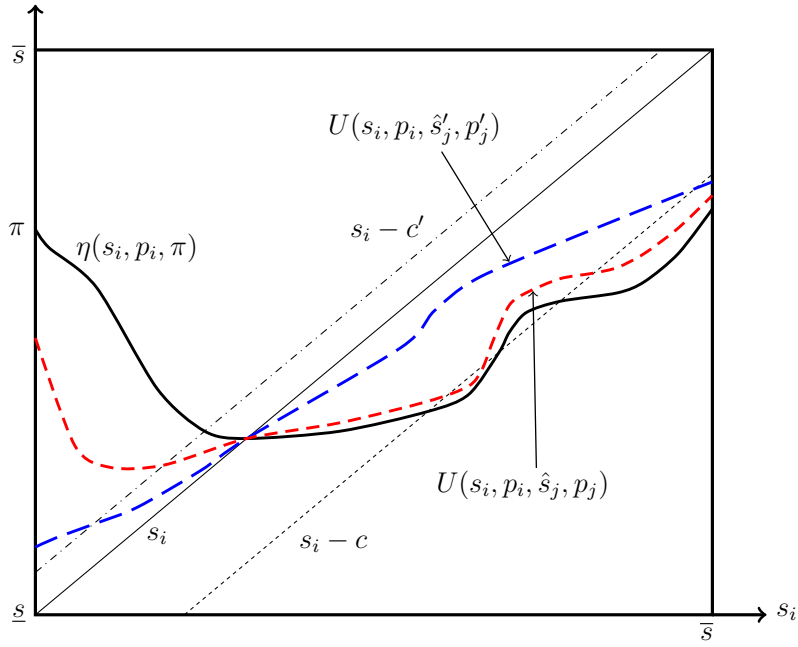
\begin{figure}

\centering
\begin{tikzpicture}[scale=0.9, transform shape, x=10cm,y=8.35cm]

\def\axiswidth{1.2pt}
\def\curvewidth{1.1pt}
\def\thinwidth{0.4pt}

\draw[line width=\axiswidth] (0,0) rectangle (1,1);

\draw[line width=\axiswidth,->] (0,1) -- (0,1.08);
\draw[line width=\axiswidth,->] (1,0) -- (1.08,0);

\node[above left] at (0,0.96) {$\bar s$};
\node[below left] at (0,0.035) {$\underline s$};
\node[below right] at (0.967,0.005) {$\bar s$};
\node[below] at (1.11,0.038) {$s_i$};

\node[left] at (0,0.6826) {$\pi$};


\draw[line width=\thinwidth] (0,0) -- (1,1);
\node at (0.18,0.14) {$s_i$};

\draw[dash pattern=on 1.5pt off 1.5pt,line width=\thinwidth] (0.22,0.0) -- (1,0.78);
\node at (0.44,0.13) {$s_i-c$};

\draw[dash dot,line width=\thinwidth] (0.0,0.075) -- (0.925,1.0);
\node at (0.52,0.684) {$s_i-c'$};

\draw[line width=\curvewidth]
plot[smooth] coordinates {
(0.00,0.6826)
(0.02,0.65)
(0.09,0.58)
(0.18,0.40)
(0.25,0.325)
(0.312,0.312)
(0.45,0.325)
(0.58,0.36)
(0.64,0.39)
(0.685,0.465)
(0.70,0.50)
(0.725,0.535)
(0.775,0.555)
(0.86,0.57)
(0.90,0.59)
(0.95,0.64)
(1,0.7186)
};
\node at (0.15,0.65) {$\eta(s_i,p_i,\pi)$};

\draw[dash pattern=on 9pt off 4pt, line width=\curvewidth, color=blue]
plot[smooth
] coordinates {
(0,0.1198)
(0.03,0.136)
(0.085,0.16)
(0.15,0.19)
(0.245,0.256)
(0.312,0.312)
(0.54,0.4731)
(0.59,0.5389)
(0.66,0.60)
(0.86,0.70)
(1,0.7665)
};

\node at (0.56,0.86) {$U(s_i,p_i,\hat{s}'_j,p'_j)$};

\draw[->,line width=0.5pt]
(0.60,0.82) -- (0.70,0.625);

\draw[dash pattern=on 5pt off 3pt, line width=\curvewidth, color=red]
plot[smooth] coordinates {
(0,0.49)
(0.058,0.3)
(0.11,0.2611)
(0.19,0.265)
(0.312,0.312)
(0.45,0.34)
(0.58,0.378)
(0.64,0.405)
(0.66,0.435)
(0.677,0.488)
(0.70,0.548)
(0.725,0.57)
(0.775,0.595)
(0.825,0.605)
(0.86,0.615)
(0.90,0.64)
(0.95,0.685)
(1,0.7425)
};

\node at (0.72,0.23) {$U(s_i,p_i,\hat s_j,p_j)$};

\draw[->,line width=0.5pt]
(0.74,0.26) -- (0.735,0.573);
\end{tikzpicture}

\caption{Sender $i$'s ``best response'' to sender $j$ (both upward biased), illustrated with $c>0>c'$, $p'_j>p_j$, and $\hat s'_j<\hat s_j<\overline s$.}
\label{fig:two}
\end{figure}

It is evident from \autoref{fig:two} that $j$'s information disclosure has very different consequences for $i$'s best response depending on the cost structure. When $c>0$ (disclosure cost), the smallest and largest solutions to \eqref{e:BR} are respectively larger than the smallest and largest single-sender thresholds.\footnote{The ordering can reverse for some intermediate solutions, as seen in the figure. As noted later, our key equilibrium conclusions hold not just for extremal solutions, but for any solution that is dynamically stable.} When $c<0$ (concealment cost, notated as $c'$ in the figure), the largest solution is smaller than the unique single-sender threshold.\footnote{Although not seen in the figure, there can be multiple solutions to \eqref{e:BR} even when $c<0$.} If $c=0$, the unique solution equals the single-sender threshold. These contrasting effects reflect the different nature of disagreement induced by the sender upon nondisclosure. With a disclosure cost, the threshold type has a \emph{higher} belief than the DM upon nondisclosure, so an expected shift of the DM's posterior toward the threshold belief makes concealment more attractive. By contrast, with a concealment cost, the threshold type has a \emph{lower} belief than the DM, so the same shift makes concealment less attractive. The logic applies not only to comparisons with the single-sender benchmark but also to any change in $p_j$ or $\hat s_j$ that makes $j$'s message more informative, as also seen in \autoref{fig:two}.

Since there can be multiple solutions to \eqref{e:BR}, in general ${\hat s_i}^{BR}(\cdot)$ is a best-response correspondence. We say that $i$'s best response \emph{increases} if the largest and smallest elements of ${\hat s_i}^{BR}(\cdot)$ both increase (weakly). We write $\hat{s_i}^{BR}(\cdot)> \hat s$ if the smallest element strictly exceeds $\hat s$.

\begin{proposition}
\label{lem:subcomp}
Assume both senders are upward biased. Any sender $i$'s best response ${\hat s_i}^{BR}(\hat s_j,p_i,p_j)$ is decreasing in $p_i$. Furthermore, let $\hat s^0_i$ denote the unique (resp., smallest) equilibrium threshold in the single-sender game with $i$ when $c\leq 0$ (resp., $c>0$). We have:
\begin{enumerate}
\item \label{subcomp1} \textsc{(independence)} \ If $c=0$, then ${\hat s_i}^{BR}(\hat s_j,p_i,p_j)=\hat s_i^0$ is independent of $\hat s_j$ and $p_j$.
\item \label{subcomp2} \textsc{(strategic complements)} \ If $c < 0$, then (i) ${\hat s_i}^{BR}(\hat s_j,p_i,p_j) \leq \hat s_i^0$, with equality if and only if $\hat s_i^0=\underline s$ or $\hat s_j=\overline s$, and (ii) ${\hat s_i}^{BR}(\hat s_j,p_i,p_j)$ increases in $\hat s_j$ and decreases in $p_j$.
\item \label{subcomp3} \textsc{(strategic substitutes)} \ If $c > 0$, then (i) ${\hat s_i}^{BR}(\hat s_j,p_i,p_j) \geq \hat s_i^0$, with equality if and only if $\hat s_i^0=\overline s$ or $\hat s_j=\overline s$, and (ii) ${\hat s_i}^{BR}(\hat s_j,p_i,p_j)$ decreases in $\hat s_j$ and increases in $p_j$.
\end{enumerate}
\end{proposition}

Since each sender's best response is monotone, existence of an equilibrium follows from Tarski's fixed point theorem. When $c < 0$ (concealment cost), strategic complementarity implies that there is a largest equilibrium, corresponding to the highest thresholds for both senders. Each sender's message in the largest equilibrium is a garbling of his message in any other equilibrium. It follows that the largest equilibrium is the least informative and the \emph{worst} for DM welfare.  Similarly, the smallest equilibrium---that with the lowest equilibrium thresholds for both senders---is the most informative and the \emph{best} in terms of the DM's welfare.

On the other hand, when $c > 0$ (disclosure cost), the two disclosure thresholds are strategic substitutes.  There is an \emph{$i$-maximal} equilibrium that maximizes sender $i$'s threshold and also minimizes sender $j$'s threshold across all equilibria.  Likewise, there is a \emph{$j$-maximal} equilibrium that minimizes sender $i$'s threshold and also maximizes sender $j$'s threshold across all equilibria.  These two equilibria are not ranked in terms of informativeness and in general cannot be welfare ranked for the DM; moreover, neither may correspond to either the best or the worst equilibrium for the DM.\footnote{As explained after \autoref{prop:single}, the senders' welfare ranking across equilibria just depends on the probability of disclosure.  When $c<0$, both senders' welfare is lowest in the largest equilibrium and highest in the smallest equilibrium. When $c>0$, sender $i$'s welfare is highest in the $i$-maximal equilibrium and lowest in the $j$-maximal equilibrium.}

The following result is derived using standard monotone comparative statics. Although stated for extremal equilibria, the comparative statics also hold for any equilibrium that is stable in the sense of adaptive dynamics \citep{Echenique2002}.

\begin{proposition}
\label{prop:welfare}
Assume both senders are upward biased. For any $i\in\{1,2\}$:
\begin{enumerate}
\item \label{welfare1} If $c\leq 0$, an increase in $p_i$ or a decrease in $c$ (higher concealment cost) weakly lowers both senders' thresholds in both the worst and best equilibria.
\item \label{welfare2} If $c>0$, an increase in $p_i$ weakly lowers sender $i$'s threshold and weakly raises sender $j$'s in both the $i$-maximal and $j$-maximal equilibria. A decrease in $c$ (lower disclosure cost) has ambiguous effects on equilibrium thresholds.
\end{enumerate}
\end{proposition}

The intuition for \autoref{prop:welfare} is as follows. \autoref{prop:singleCS} indicates that both an increase in $p_i$ and a decrease in $c$ lower sender $i$'s best response, holding fixed the other sender's strategy. The equilibrium effects then depend on the cost structure. Under concealment cost, the strategic complementarity established in \autoref{lem:subcomp} implies that a sender's increased disclosure promotes the other's. Thus, for changes in $p_i$, the direct effect on sender $i$ and the indirect effect on $j$ reinforce each other. For changes in $c$, the direct effects on both senders and the indirect effects all go in the same direction. Under disclosure cost, the strategic substitution creates a tension: a sender who discloses more crowds out the other's disclosure. For changes in $p_i$, which directly affect only sender $i$'s best response, the direct effect dominates and the equilibrium comparative statics are unambiguous. But since the message cost is common to both senders, a change in $c$ has a direct effect of shifting both best responses in the same direction, while inducing offsetting indirect effects, and the net effect on equilibrium thresholds is ambiguous. In particular, a reduction in the disclosure cost can (strictly) reduce DM welfare, as illustrated in \appendixref{app:welfare-example}.\footnote{If we had instead allowed for a sender-specific cost $c_i$, then a decrease in the disclosure cost $c_i>0$ would have the same unambiguous overall effects as an increase in $p_i$. As symmetric changes in the message cost for both senders are quite natural (e.g., when message costs are viewed as an institutional or policy parameter), we think it is interesting that a common decrease in disclosure cost has an ambiguous effect.}

One can view the single-sender game with sender $i$ as a two-sender game where the other sender is never informed (i.e., $p_j=0$). A comparison of the single-sender and two-sender games is then a corollary of \autoref{lem:subcomp} and \autoref{prop:welfare}, as follows.

\begin{corollary}
\label{cor:main}
Assume both senders are upward biased. Let $\hat s^0_i$ denote the unique (resp., smallest) equilibrium threshold in the single-sender game with $i$ when $c\leq 0$ (resp., $c>0$).
\begin{enumerate}
\item \label{main1}
If $c=0$, then $(\hat s^0_1,\hat s^0_2)$ is the unique equilibrium in the two-sender game. The DM's welfare is strictly higher in the two-sender game than in either single-sender game.
\item \label{main2}
If $c < 0$, then every equilibrium in the two-sender game is weakly smaller than $(\hat s^0_1,\hat s^0_2)$, with equality if and only if $\hat s^0_1=\hat s^0_2=\underline s$. The DM's welfare is strictly higher in any equilibrium of the two-sender game than in either single-sender game.
\item \label{main3}
If $c > 0$, then every equilibrium in the two-sender game is weakly larger than $(\hat s^0_1,\hat s^0_2)$, with equality if and only if $\hat s^0_1=\hat s^0_2=\overline s$. The DM's welfare in the best equilibrium of the two-sender game may be higher or lower than in the best equilibrium of either single-sender game.
\end{enumerate}
\end{corollary}

Part \ref{main1} of \autoref{cor:main} follows from part \ref{subcomp1} of \autoref{lem:subcomp}. When $c=0$, the best response of each sender is to use the same disclosure threshold as in the single-sender setting, regardless of the other sender's strategy. Since the DM receives two messages instead of just one, and the probability distribution of these messages remain the same as in the single-sender game, she is better off when facing both senders than when facing either sender alone.

Part \ref{main2} of \autoref{cor:main} can be obtained by considering the worst equilibrium of the two-sender game. When there is a concealment cost ($c<0$), let ${\mathbf{s}}^*(p_i,p_j)$ represent the vector of disclosure thresholds in the worst equilibrium. \autoref{prop:welfare} implies that $s^*_i(p_i,p_j) \le s^*_i(p_i,0)=\hat{s}^0_i$ and $s^*_j(p_i,p_j) \le s^*_j(0,p_j)=\hat{s}^0_j$. Thus $\mathbf{s}^*(p_i,p_j)$ is weakly smaller than $(\hat{s}^0_i,\hat{s}^0_j)$, and hence every equilibrium is weakly smaller than $(\hat{s}^0_i,\hat{s}^0_j)$. It follows that the DM's welfare is higher than in the unique equilibrium of the single-sender game with either sender. This higher welfare is due to both a direct effect of receiving information from an additional sender, and an indirect effect wherein each sender is now disclosing more than in the single-sender setting.

Finally, when there is a disclosure cost ($c>0$), let ${\mathbf{s}}^{i*}(p_i,p_j)$ represent the $i$-maximal equilibrium and $\mathbf{s}^{j*}(p_i,p_j)$ represent the $j$-maximal equilibrium. In any equilibrium, $i$'s threshold is at least as large as $s^{j*}_i(p_i,p_j) \ge s^{j*}_i(p_i,0) =\hat{s}^0_i$, where the inequality is by part \ref{welfare2} of \autoref{prop:welfare}. Analogously, sender $j$'s threshold in any equilibrium is at least as large as $s^{i*}_j(p_i,p_j) \ge s^{i*}_j(0,p_j)=\hat{s}^0_j$. Thus, both senders are (weakly) less informative than in the DM's best equilibrium of the single-sender game. The overall welfare comparison between the two-sender game and the single-sender game is generally ambiguous. While adding a second sender has a direct effect of increasing the DM's information, there is an adverse indirect effect due to the strategic substitution in disclosure of the other sender. It is possible the net effect can (strictly) reduce the DM's welfare---even when the two senders have opposing biases, which is often thought to particularly promote information disclosure. We provide an explicit example with a familiar quadratic loss function for the DM in \appendixref{app:welfare-example}.

\paragraph{Discussion.} The assumption of conditionally independent signals is important for our analytical methodology. Without it, we cannot directly apply the IVP result from \citet*{KLS21-IVP}. If signals were conditionally correlated, then upon nondisclosure, sender $i$ and the DM would disagree not only on the probability assessment of the states but also on the experiment corresponding to sender $j$'s message. A general analysis appears intractable. We illustrate in \appendixref{app:correlated} how our conclusions change under a very different information structure: perfectly correlated signals.\footnote{\citet{BM13} and \citet{EF19} both study multi-sender disclosure with perfectly correlated signals. \citet{BM13} assume no message costs ($c=0$) and find that adding senders can reduce DM welfare, but only if senders have non-monotonic preferences. \citet{EF19} assume a disclosure cost ($c>0$) and find that the DM can sometimes benefit from barring one sender.} 
Another important assumption to apply the IVP result is that each sender has linear preferences. \appendixref{app:nonlinear} discusses how our result under $c = 0$ extends to non-linear preferences and how our results under $c \ne 0$ may or may not hold under non-linear preferences.

\section{Extensions}
\label{sec:extensions}

\subsection{Many Senders}
\label{sec:manysenders}

Our results readily generalize to any finite number of senders. Suppose in addition to senders $i$ and $j$, there are $K$ other senders, all of whom receive conditionally independent signals and simultaneously send messages to the DM. Let $\mathbf{m}$ represent the collection of these $K$ messages. Then, sender $i$'s posterior belief given his own signal $s_i$, sender $j$'s message $m_j$, and the $K$ other senders' messages $\mathbf{m}$ is $\beta(m_j,\mathbf{m},s_i)=\beta(m_j,\beta(\mathbf{m},s_i))$, where we abuse notation and use $\beta(\cdot,\mu_i)$ to  denote $i$'s posterior from updating belief $\mu_i$ on any collection of messages, suppressing the parameters needed to interpret those messages. The foregoing equality reflects that updating on $m_j$ and $\mathbf{m}$ jointly is equivalent to updating first on $\mathbf{m}$ and then on $m_j$. The DM's belief given the $K$ senders' messages $\mathbf{m}$ and given nondisclosure by sender $i$ is $\eta(\hat s_i,p_i,\beta(\mathbf{m},\pi))$. Thus, by the law of iterated expectations, the expected payoff for sender $i$ from concealing his signal is
\begin{equation*}
\E \Big[ \E_{\hat s_j,p_j} \big[ T(\beta(m_j,\beta(\mathbf{m},s_i)),\eta(\hat s_i,p_i,\beta(\mathbf{m},\pi)),\beta(\mathbf{m},s_i)) \mid \mathbf{m}\big] \Big],
\end{equation*}
where the inside conditional expectation (given $\mathbf{m}$) is taken over the distribution of $m_j$, while the outside expectation is taken over the distribution of $\mathbf{m}$ generated from the equilibrium strategies of the $K$ senders.  

Given any message profile $\mathbf{m}$, the transformation $T(\cdot)$ is the same as that in the two-sender case, with the common prior $\pi$ replaced by $\beta(\mathbf{m},\pi)$. 
Moreover, because the order of updating on $\mathbf{m}$ and on $i$'s nondisclosure can be interchanged, the direction of disagreement between sender $i$'s belief and the DM's nondisclosure belief is the same for every realization of $\mathbf{m}$. Since our results hold for any $\pi$, the logic of strategic substitution or strategic complementarity continues to apply. In particular, when $c < 0$, sender $i$'s expectation of the DM's posterior conditional on $\mathbf{m}$ increases in $\hat s_j$ and decreases in $p_j$ for any $\mathbf{m}$. 
Consequently, $i$'s expected payoff from nondisclosure also increases in $\hat s_j$ and decreases in $p_j$. Thus disclosures by any two senders are strategic complements. Similarly, with a disclosure cost ($c>0$), disclosures are strategic substitutes.

It follows from these observations that when there is either no message cost or a concealment cost ($c \le 0$), the DM always benefits from having more senders to supply her with information. When there is a disclosure cost ($c > 0$), on the other hand, an increase in the number of senders has ambiguous effects on each sender's disclosure threshold, and can lead to either an increase or decrease in the DM's welfare.

\subsection{Sequential Reporting}
\label{sec:sequential}

The key insight from our analysis of simultaneous disclosure extends to sequential disclosure. For concreteness, consider a two-sender game in which both senders are upward biased but disclosure is sequential: sender 1 reports first and his message $m_1$ is made public to both the DM and sender 2 before sender 2 submits his report. Sender 2 now effectively faces a single-sender problem in which he and the DM share a common prior, say $\beta(m_1,\pi)$, which is an equilibrium object. \autoref{prop:singleCS} implies that sender 2 will adopt a disclosure threshold $\hat s_2^0$ that decreases in $p_2$.

Consider now the disclosure decision of sender 1 when the DM conjectures that he is using a disclosure threshold $\hat s_1$, with corresponding nondisclosure belief $\eta(\hat s_1,p_1,\pi)$. If sender 1 discloses his signal $s_1$, his expectation of the DM's posterior belief is simply $s_1$. If he chooses nondisclosure, his expectation is
$
\E_{\hat s^0_2,p_2}[T(\beta(m_2,s_1),\eta(\hat s_1,p_1,\pi),s_1)]$.
Since sender 2 discloses more when he is better informed, a higher $p_2$ makes the message $m_2$ more informative, both directly through a higher probability of sender 2 getting a signal and indirectly through a lower disclosure threshold $\hat s_2^0$. IVP implies that sender 1 expects the DM's belief to move away from $\eta(\hat s_1,p_1,\pi)$ toward $s_1$. The same logic that establishes \autoref{prop:welfare} therefore gives the following result, whose proof is omitted.

\begin{proposition}
\label{prop:sequential}
Consider sequential disclosure and assume sender 1, the first mover, is upward biased. If $c>0$ (resp., $c<0$), a higher $p_2$ weakly increases 
(resp., weakly reduces) 
the equilibrium disclosure threshold of sender 1 in the equilibria with the highest and lowest thresholds for sender 1.
\end{proposition}

An immediate corollary to \autoref{prop:sequential} is that, given a concealment cost, the first sender discloses more than he does in a single-sender setting. As a result, the DM is always better off in a sequential game than with the first sender alone. On the other hand, a welfare comparison between the sequential- and simultaneous-move games is generally ambiguous.

We also note that if $c=0$ the irrelevance result still holds for sender 1: his disclosure threshold under sequential reporting coincides with that from his single-sender problem. The disclosure threshold chosen by the second sender, however, depends on sender 1's message; as it depends nonlinearly on the belief induced by that message, it can be higher or lower---even on average, ex ante---than sender 2's single-sender disclosure threshold.

\subsection{Changes in Signal Precision}
\label{sec:precision}

Our main results---\autoref{lem:subcomp} and \autoref{prop:welfare}---consider changes in a sender's informativeness via the probability of being informed, $p_i$. We now explain how these insights extend to the intensive margin: the quality of a sender's information conditional on being informed.

Suppose that when informed, sender $i$'s signal or private belief is drawn from a distribution indexed by a precision parameter $\rho\in \Reals$. It is convenient to now characterize the information structure directly by the \emph{unconditional} distribution of beliefs, $F^\rho(s)$ with density $f^\rho(s)$. Combined with the prior $\pi$, Bayes' rule uniquely determines the state-contingent signal distributions. We assume differentiability of $F^\rho$ and $f^\rho$ in $\rho$, and that higher precision yields a \emph{rotation} of the signal distribution around the prior in the sense of \citet{JohnsonMyatt2006}: $\sign\left[\partial F^{\rho}(s) / \partial \rho\right]=\sign\left[\pi-s\right]$. This implies, in particular, that a higher $\rho$ corresponds to a mean-preserving spread of $F^\rho$ (and hence more information in the Blackwell sense).

The comparative statics of precision on strategic behavior operates through the same channel as the probability of being informed. Recall that by \autoref{prop:singleCS}, an increase in the probability of being informed, $p_i$, makes sender $i$ disclose more (lowers $i$'s equilibrium threshold, given upward bias); this is because increasing $p_i$ makes the nondisclosure belief less favorable, making concealment less attractive. If increasing sender $i$'s precision has the same directional effect on the nondisclosure belief, then that logic carries over, and consequently the comparative statics of \autoref{lem:subcomp} and \autoref{prop:welfare} also follow. So, the key question is: does increasing precision lower the nondisclosure belief?

We can study that by considering the effects of local changes in precision. Adapting the nondisclosure belief from \autoref{e:NDbelief}---dropping the dependence on $\pi$ and $p_i$ to ease notation, and instead adding the dependence on $\rho$---and differentiating, we obtain
\begin{equation}
\label{e:precision-derivative}
\sign\left[\frac{\partial}{\partial \rho} \eta(\hat s,\rho)\right] = \sign\left[(\hat{s} - \eta(\hat s,\rho))\frac{\partial}{\partial \rho} F^\rho(\hat{s}) - \int_{\underline{s}}^{\hat{s}} \frac{\partial}{\partial \rho} F^\rho(s) \mathrm ds\right].
\end{equation}
(\autoref{rem:precision-deriv} in the Appendix details the derivation.) The integral on the right-hand side of \autoref{e:precision-derivative} is nonnegative by the mean-preserving spread of a higher $\rho$. So it is the other term that needs attention.

Consider first a concealment cost ($c < 0$). We know that in this case sender $i$'s equilibrium disclosure threshold $\hat s_i$ satisfies $\hat{s}_i < \eta(\hat s_i,\cdot)$, and hence by \autoref{lem:NDbelief}, $\hat s_i<\pi$. Consequently, the first term on the right-hand side of \eqref{e:precision-derivative} is negative because $\frac{\partial}{\partial \rho} F^{\rho}(\hat{s_i}) > 0$ by the rotation property.  Both terms of the right-hand side of \eqref{e:precision-derivative} thus work in the same direction:  increasing precision lowers the nondisclosure belief. The logic of strategic complements under concealment cost therefore extends.

Turning to a {disclosure cost} ($c > 0$), sender $i$'s equilibrium threshold now satisfies $\hat{s}_i > \eta(\hat s_i,\rho)$. If $\hat{s_i} > \pi$ (as is assured, for example, if $c$ is large enough) then the rotation property implies $\frac{\partial}{\partial \rho} F^{\rho}(\hat{s_i}) < 0$, and it follows that \eqref{e:precision-derivative} is negative. However, in general it is possible that $\hat s_i<\pi$. Then $\frac{\partial}{\partial \rho} F^{\rho}(\hat{s_i}) > 0$ and so $\left(\hat{s} - \eta(\hat s_i,\rho)\right)\frac{\partial F^\rho(\hat{s})}{\partial \rho}>0$, which means we have offsetting effects in the right-hand side of \eqref{e:precision-derivative}. A further assumption is needed to ensure that the nondisclosure belief decreases. It is sufficient that the improvement in precision satisfies a \emph{decreasing likelihood ratio} property on the lower half of the private-belief space: for $s<s'<\pi$, the ratio $f^\rho(s')/f^\rho(s)$ is decreasing in $\rho$. This property holds in common parametric families of distributions.\footnote{For example, with a symmetric Beta distribution of beliefs (so the prior $\pi=1/2$) with parameter $1/\rho>0$, we have $f^\rho(s) \propto[s(1-s)]^{1/\rho-1}$, and hence $$\frac{f^\rho\left(s^{\prime}\right)}{f^\rho(s)}=\left[\frac{s^{\prime}\left(1-s^{\prime}\right)}{s(1-s)}\right]^{1 / \rho-1}.$$
For $s<s^{\prime}<1 / 2$, the term in square brackets exceeds 1, which implies the desired monotonicity in $\rho$. One can also check that the monotonicity holds for the distribution of beliefs derived from an underlying Normal signal structure, i.e., where the belief $s(y)$ is computed from a primitive signal $y = \omega + \varepsilon$, with $\varepsilon \sim \mathcal{N}(0,1/\rho)$.} Intuitively, it ensures that an increase in precision shifts probability mass relatively more towards extreme signals/beliefs rather than ones near the threshold $\hat s_i$, and so the downward pull of the integral term on the right-hand side of \eqref{e:precision-derivative} outweighs the upward pull of the first term when $\hat s_i<\pi$; see \autoref{rem:DLR} in the Appendix for a formal verification. Increasing precision then lowers the nondisclosure belief no matter the disclosure threshold, and consequently, the logic of strategic substitutes under disclosure cost extends.

\subsection{Uncertain Bias}
\label{sec:uncertain-bias}

Our main themes about strategic complements or substitutes also hold when senders' biases are not common knowledge. Let us sketch why.

Suppose that each sender's direction of bias is drawn independently and is his private information: with probability $\lambda$ a sender is upward biased ($b_i=1$), a $U$-type, and with probability $1-\lambda$ is downward biased ($b_i=-1$), a $D$-type. The parameter $\lambda$ is common knowledge; it is straightforward to allow it to be a sender-specific $\lambda_i$ at the cost of more notation. The rest of the model is unchanged from \autoref{sec:model}.

Consider first the single-sender benchmark. The $U$-type uses a threshold $s_u$, disclosing signals above $s_u$, while the $D$- type uses a threshold $s_d$, disclosing signals below $s_d$. Let $\eta(s_u, s_d)$ denote the DM's nondisclosure belief, which depends on both thresholds.\footnote{Analogously to \autoref{e:NDbelief}, we now have
\[\eta(s_u, s_d) := \frac{(1-p)\pi + \lambda p F_\pi(s_u)\E_\pi[s \mid s < s_u] + (1-\lambda)p(1-F_\pi(s_d))\E_\pi[s \mid s > s_d]}{(1-p) + \lambda p F_\pi(s_u) + (1-\lambda)p(1-F_\pi(s_d))}.\]} The indifference conditions (assuming interior thresholds, for simplicity) are:
\[s_u - c = \eta(s_u, s_d) \quad \text{and} \quad -s_d - c = -\eta(s_u, s_d).\]
Together these imply $s_u = s_d + 2c$, and equilibrium reduces to a fixed point of one variable. When $c > 0$, we have $s_u > s_d$, implying an interval $(s_d, s_u)$ of signals that neither type discloses; when $c < 0$, we have $s_u < s_d$, and the disclosure regions overlap on $(s_u, s_d)$.

Now consider what happens when the DM receives some \emph{exogenous} information simultaneously, or subsequently, to the single sender's disclosure. This conveys the key intuition of the effects of a second sender. The IVP logic driving \autoref{lem:rotation} still applies, but with a twist: each sender type expects a different direction of DM's belief revision, yet both types' incentives move in the same direction. Specifically, consider first a disclosure cost ($c > 0$). As the $U$-type has $s_u > \eta$, IVP implies that the threshold $U$-type expects information to pull the DM's nondisclosure belief upward toward $s_u$. Conversely, the threshold $D$-type expects a downward revision towards $s_d$. Crucially though, given their opposing biases, this means that both types expect the additional information after concealment to generate (on average) a favorable revision. Consequently, concealment becomes more attractive for each type, while there is no change in the expected payoff from disclosure. This is the logic underlying strategic substitutes in our baseline model.

With a concealment cost ($c < 0$), matters flip. Now, since $s_u < \eta(s_u,s_d)$ and $s_d > \eta(s_u,s_d)$, the $U$-type expects the additional information to pull the DM's nondisclosure belief downward, while the $D$-type expects an upward revision. So both types expect unfavorable revisions upon nondisclosure, making concealment less attractive for each. This is the logic underlying strategic complements in the baseline model. 

The upshot is that despite uncertainty about a sender's bias, the cost structure determines the impact of strategic interaction just as earlier---even though the two sender types expect the DM's belief revision to go in opposite directions after nondisclosure.

\section{Conclusion}
\label{sec:conclusion}

This paper has studied multi-sender persuasion when senders can reveal or conceal private information at a cost. Our central insight is that the nature of these costs---whether they are, on net, attached to disclosure or to concealment---fundamentally shapes the strategic interaction among senders.

When concealment is costly, senders' disclosures are strategic complements: each sender discloses more aggressively when others do the same. When disclosure is costly, disclosures are instead strategic substitutes. The common mechanism owes to disagreement between the threshold sender type and the DM upon nondisclosure, and that each sender expects (on average) others' messages to reduce that disagreement gap. Under a concealment cost, the threshold sender type has a less favorable belief than the DM's nondisclosure belief; the gap offsets the cost of concealment. This type thus expects additional senders' messages to move the DM's nondisclosure belief less favorably, which makes concealment less attractive. The logic reverses under a disclosure cost, making disclosure less attractive.

These findings imply that in the model we study, the conventional wisdom that competition among information sources promotes disclosure and benefits decision makers holds under no message costs or concealment costs, but can fail under disclosure costs. With disclosure costs, adding senders can actually reduce DM welfare. Similarly, lowering those costs---while beneficial in a single-sender setting---can reduce DM welfare with multiple senders. While these welfare effects may be reminiscent of free-rider considerations, the current logic is distinct and stems from how the DM interprets silence; the contrast with free-riding is sharpened by noting the DM's indirect benefit from additional senders under concealment cost. The upshot is that policymakers seeking to enhance information provision should attend not just to the number of information sources but to the underlying cost structure governing their communication.

Disclosure costs are likely prominent when information provision requires significant preparation or legal review, or when it risks revealing proprietary knowledge; salient contexts include litigation discovery or financial disclosures that may inform competitors. Concealment costs may dominate when suppression risks ex-post penalties or reputational damage, or when legal duties to disclose apply; relevant settings include Brady obligations in prosecution, fiduciary duties in financial advising, or duty-to-disclose regimes in securities and real estate. Our analysis suggests that expanding the number of information sources is unambiguously beneficial in the latter settings but may backfire in the former.

Finally, while we have framed our discussion in terms of multiple strategic senders, our analysis applies just as well if there is only one strategic sender and the DM receives additional information from an arbitrary other source (cf.~\autoref{sec:uncertain-bias}). So, for instance, when the sender bears a disclosure cost, the DM can be hurt by the option to freely acquire a limited amount of information; she may prefer to tie her hands ex ante to not do so.

\appendix
\section{Appendix}
\label{sec:appendix}

\subsection{Proofs}

\begin{proof}[Proof of \autoref{lem:NDbelief}]
Partially differentiating \eqref{e:NDbelief} with respect to the first argument yields
\begin{align*}
\frac{\partial \eta(\hat s,p,\pi)}{\partial \hat s} &= \frac{-p(1-p)f_{\pi}(\hat s)}{(1-p+pF_{\pi}(\hat s))^2}\left(\pi - \E_{\pi}[s \mid s<\hat s]\right) +
\frac{pF_{\pi}(\hat s)}{1-p+pF_{\pi}(\hat s)}\frac{f_{\pi}(\hat s)}{F_{\pi}(\hat s)} \left(\hat s-\E_{\pi}[s \mid s<\hat s]\right) \\
&=\frac{pf_{\pi}(\hat s)}{1-p+pF_{\pi}(\hat s)}\left(
\frac{-(1-p)}{1-p+pF_{\pi}(\hat s)}\left(\pi-\E_{\pi}[s \mid s<\hat s]\right) + \left(\hat s-\E_{\pi}[s \mid s<\hat s]\right)\right) \\
&=\frac{pf_{\pi}(\hat s)}{1-p+pF_{\pi}(\hat s)}\left(\hat s-\eta(\hat s,p,\pi)\right).
\end{align*}
Hence, $\sign\left[\partial \eta(\hat s,p,\pi)/\partial \hat s\right]=\sign\left[\hat s-\eta(\hat s,p,\pi)\right]$. Part \ref{NDbelief1} of the lemma follows from the observation that for any $p$ and $\pi$, we have $\eta(\underline s,p,\pi)=\eta(\overline s,p,\pi)=\pi\in (\underline s,\overline s)$.

Partially differentiating \eqref{e:NDbelief} with respect to the second argument and simplifying yields
$$\frac{\partial \eta(\hat s,p,\pi)}{\partial p}=\frac{F_\pi(\hat s)\left(\E_{\pi}[s \mid s<\hat s]-\pi\right)}{\left(1-p+pF_\pi(\hat s)\right)^2},$$
which proves part \ref{NDbelief2} because $\E_{\pi}[s \mid s<\hat s]<\pi$ if and only if $\hat s<\overline s$, and $F_\pi(\hat s)>0$ if and only if $\hat s>\underline s$.
\end{proof}

\begin{proof}[Proof of \autoref{prop:single}] We prove each part in turn.

\emph{Part \ref{single1}:} Consider an interior threshold $\hat s_i^0 \in (\underline s, \overline s)$. The sender with signal $\hat s_i^0$ must be indifferent between disclosing (with payoff $\hat s_i^0 - c$) and concealing (with payoff $\eta(\hat s_i^0, p_i, \pi)$), yielding condition (i). Now consider the boundary case $\hat s_i^0 = \underline s$. Then the nondisclosure belief is the prior $\pi$, and since even the lowest type is willing to disclose, we must have $\pi \leq \underline s - c$. In the other boundary case, $\hat s_i^0 = \overline s$, the nondisclosure belief is again the prior $\pi$, and since even the highest type is willing to conceal, we must have $\pi \geq \overline s - c$. Conversely, any threshold satisfying these conditions constitutes an equilibrium by standard verification.

\emph{Part \ref{single2}:} Assume $c\leq 0$. Then the boundary $\overline s$ cannot be an equilibrium threshold, as $\pi < \overline s - c$. So equilibrium thresholds are either interior or $\underline s$. If $\underline s - c < \pi=\eta(\underline s,p_i,\pi)$, then any equilibrium must have an interior threshold $\hat s$ and hence satisfy $\eta(\hat s,p_i,\pi)=\hat s-c\geq \hat s$. Since $\eta(\cdot,p_i,\pi)$ is strictly decreasing in the region where $\eta(s,p_i,\pi) \geq s$ (by \autoref{lem:NDbelief}), and the line $s - c$ is increasing in $s$, there is only one solution to $\eta(\hat s,p_i,\pi) = \hat s-c$. If instead $\underline s - c \geq \pi$, then $\underline s$ is an equilibrium threshold; moreover, $s>\underline s$  implies $s - c>\eta(s,p_i,\pi)$---because $\eta(\cdot,p_i,\pi)$ is quasiconvex by \autoref{lem:NDbelief} and $\eta(\underline s,p_i,\pi)=\eta(\overline s,p_i,\pi)$---and hence there is no other equilibrium. Finally, when $c = 0$, we have $\underline s - c = \underline s < \pi$, so the unique equilibrium is interior; it is the fixed point of $\eta(\cdot, p_i, \pi)$.

\emph{Part \ref{single3}:} When $c > 0$, an interior equilibrium threshold $\hat s$ requires $\eta(\hat s, p_i, \pi) = \hat s - c < \hat s$, placing solutions in the region where $\eta$ is strictly increasing (by \autoref{lem:NDbelief}). As noted in the ``validity'' claim in \autoref{fn:fig-ND}, parameters can be chosen so that $\eta(\cdot,p_i,\pi)$ has slope exceeding one at some points in this region; for sufficiently small $c > 0$, the line $s - c$ then crosses $\eta(s,p_i,\pi)$ multiple times, as illustrated in \autoref{fig:single}. It remains to explain the validity claim in \autoref{fn:fig-ND}. Given any such quasiconvex function $\psi$, set $\pi = \psi(\underline s)$. The nondisclosure belief formula \eqref{e:NDbelief} then becomes a first-order ordinary differential equation in $F_\pi$ with boundary condition $F_\pi(\underline s)=0$; quasiconvexity of $\psi$ ensures the solution is monotone. The probability $p$ is chosen to scale the solution so that $F_\pi(\overline s) = 1$.
\end{proof}

\begin{proof}[Proof of \autoref{prop:singleCS}]
We prove only the first part of the proposition; the second part is analogous and omitted. Fix any $p_i>\tilde{p}_i$ and let $\hat s_i^0$ and $\tilde s_i^0$ denote the corresponding highest equilibrium thresholds. Suppose, to contradiction, that $\hat s_i^0>\tilde s_i^0$. Since $\tilde s_i^0$ is the highest equilibrium threshold at $\tilde p_i$, it holds for any $\hat s>\tilde s_i^0$ that $\eta(\hat s,\tilde p_i,\pi)< \hat s-c$. But $\eta(\hat{s},p,\pi)$ is weakly decreasing in $p$ (\autoref{lem:NDbelief}), so for any $\hat s>\tilde s_i^0$, we have $\eta(\hat s,p_i,\pi)< \hat s-c$. This implies $\hat s_i^0\leq \tilde s_i^0$, a contradiction. A similar argument establishes the result for the lowest equilibrium threshold.
\end{proof}

\begin{proof}[Proof of \autoref{lem:threshold}]
Fix any equilibrium and any sender $i$. Let $j\neq i$. It suffices to show that the difference in $i$'s expected payoff from disclosing versus concealing is strictly increasing in $s_i$. The expected payoff from disclosing is $s_i-c$, because the expected posterior of the DM after $i$ discloses $s_i$ is $s_i$. Denote the expected payoff from concealing as $\E [\beta_{DM}(m_j,m_i=\phi)]$, where \mbox{$\beta_{DM}(m_j,m_i=\phi)$} is the DM's equilibrium belief following any message $m_j$ and nondisclosure by $i$, and the expectation is over $m_j$ given $i$'s beliefs under $s_i$. Because $m_j$ is uncorrelated with $s_i$ conditional on the state, and $i$'s belief about the state given signal $s_i$ is $s_i$, we have
$$\E \big[\beta_{DM}(m_j,m_i=\phi)\big]=s_i \E\big[\beta_{DM}(m_j,m_i=\phi) \mid  \omega=1\big] + (1-s_i) \E\big[\beta_{DM}(m_j,m_i=\phi) \mid \omega=0\big].$$
The derivative of the right-hand side above with respect to $s_i$ is strictly less than one because \mbox{$\E\big[\beta_{DM}(m_j,m_i=\phi) \mid \omega=1\big]<1$}, as beliefs lie in $[0,1]$ and $m_j$ cannot perfectly reveal the state. Therefore, the payoff difference $(s_i-c)-\E\big[\beta_{DM}(m_j,m_i=\phi)\big]$ is strictly increasing in $s_i$.
\end{proof}

\begin{proof}[Proof of \autoref{lem:rotation}]
Given any $p_j$, $p_i$, and $\hat s_j$, \autoref{e:U} tells us that $U(s_i,p_i,\hat s_j,p_j)$---which, recall, is shorthand for $U(s_i,s_i,p_i,\hat s_j,p_j)$---is $i$'s expectation of the DM's belief under interim beliefs $s_i$ for $i$ and $\eta(s_i,p_i,\pi)$ for the DM. 

Part \ref{rotation1} of the lemma follows directly from the IVP theorem of \citet*[Theorem 1]{KLS21-IVP}. Note that the theorem's ordering assumptions---likelihood-ratio ordered priors and MLRP experiments---are automatically satisfied with a binary state, as in the current model. The strict inequalities in the consequents of part \ref{rotation1} hold because $j$'s message is not fully informative of the state. The weak inequalities are strict if and only if $j$'s message is informative, i.e., if and only if $\hat s_j<\overline s$.

For part \ref{rotation2} of the lemma, observe that the message under $(\hat s_j, p_j)$ is a garbling of the message under $(\hat s'_j, p'_j)$ whenever $p'_j \geq p_j$ and $\hat s'_j \leq \hat s_j$. To see this, note that any message $m'_j$ under $(\hat s'_j,p'_j)$ can be garbled to produce message $m_j$ as follows:
\begin{equation*}
m_j =
\begin{cases}
\phi & \mbox{if } m'_j=\phi \mbox{ or } m'_j \in [\hat s'_j,\hat s_j),\\
m'_j \mbox{ with prob. } {p_j}/{p'_j} \mbox{ or } \phi \mbox{ with prob. } 1-{p_j}/{p'_j} &\mbox{if } m'_j \ge \hat s_j.
\end{cases}
\end{equation*}
In each state, the distribution of $m_j$ as constructed is the same as the distribution of sender $j$'s message under $(\hat s_j,p_j)$. The IVP theorem \citep*[Theorem 1]{KLS21-IVP} then implies that $i$'s expected posterior under the more informative $(\hat s'_j,p'_j)$ is closer to $s_i$ than under $(\hat s_j,p_j)$. Equality holds if and only if the two messages are equally informative, which occurs when $\hat s'_j = \hat s_j$ and either $p'_j = p_j$ or $\hat s_j = \overline{s}$.
\end{proof}

\begin{proof}[Proof of \autoref{lem:subcomp}]
Recall the transformation $T(\beta_i,\mu_{DM},\mu_i)$ from \eqref{e:transform}. It is increasing in $\mu_{DM}$. \autoref{lem:NDbelief} shows that $\eta(s_i,p_i,\pi)$ is decreasing in $p_i$. Hence, sender $i$'s payoff from concealing any candidate threshold signal $s_i$ (with the threshold correctly anticipated by the DM) is
$$U(s_i,p_i,\hat s_j,p_j)=\E_{\hat s_j,p_j}\bigl[T(
\beta(m_j,\hat s_j,p_j,s_i),\eta(s_i,p_i,\pi),s_i) \mid s_i\bigr],$$ 
which is decreasing in $p_i$. The payoff from disclosure, $s_i-c$, does not depend on $p_i$. Therefore, an increase in $p_i$ makes concealment less attractive at any candidate threshold. Following the same argument as in the proof of \autoref{prop:singleCS}, if a threshold is not a best response at some $p_i$, then it is not at any higher $p_i$, and hence the largest and smallest best-response thresholds must decrease in $p_i$.

We next turn to the enumerated parts of the proposition. Let $\hat s_i^0$ be the unique (if $c\leq 0$) or smallest (if $c>0$) equilibrium threshold in the single-sender game with $i$. 

\emph{Part \ref{subcomp1}:} Consider  $c=0$. Then $\hat s_i^0$ is the fixed point of $\eta(\cdot,p_i,\pi)$. \autoref{lem:rotation} (part \ref{rotation1}) implies $U(\hat s_i^0,p_i,\hat s_j,p_j) = \hat s_i^0$, and so $\hat s_i^0\in {\hat s_i}^{BR}(\hat s_j,p_i,p_j)$. We claim there is no other best-response threshold. Suppose, to contradiction, that $\hat s' > \hat s_i^0$ and $\hat s' \in {\hat s_i}^{BR}(\hat s_j,p_i,p_j)$. \autoref{lem:NDbelief}  implies $\eta(\hat s',p_i,\pi) < \hat s'$, and then \autoref{lem:rotation} (part \ref{rotation1}) implies $U(\hat s',p_i,\hat s_j,p_j) < \hat s'$. Therefore, sender $i$ strictly prefers disclosure at signal $\hat s'$, a contradiction. A symmetric argument shows $\hat s' < \hat s_i^0$ implies $\hat s' \notin {\hat s_i}^{BR}(\hat s_j,p_i,p_j)$.

\emph{Part \ref{subcomp2}:} Now consider $c < 0$.  For any $s_i>\hat s^0_i$, the single-sender equilibrium condition implies $\eta(s_i,p_i,\pi)< s_i-c$. So $U(s_i,p_i,\hat s_j,p_j)\leq \max\{s_i,\eta(s_i,p_i,\pi)\}<s_i-c$, where the first inequality is by \autoref{lem:rotation} (part \ref{rotation1}). Hence $s_i\notin {\hat s_i}^{BR}(\hat s_j,p_i,p_j)$. Moreover, $\hat s^0_i \in {\hat s_i}^{BR}(\hat s_j,p_i,p_j)$ if and only if $\hat s^0_i=\underline s$ or $\hat s_j=\overline s$; for otherwise, $\hat s_i^0 < U(\hat s^0_i,p_i,\hat s_j,p_j)<\eta(\hat s_i^0,p_i,\pi)=\hat s_i^0-c$ (the inequalities are by part \ref{rotation1} of \autoref{lem:rotation}, and the equality by the single-sender equilibrium condition), which contradicts $\hat s^0_i \in {\hat s_i}^{BR}(\hat s_j,p_i,p_j)$. This proves part \ref{subcomp2}(i) of the proposition. For part \ref{subcomp2}(ii), note that 
for any $s_i\in {\hat s_i}^{BR}(\hat s_j,p_i,p_j)$, we have $s_i\leq \hat s^0_i$ (as just shown), and hence $\eta(s_i,p_i,\pi) > s_i$ by \autoref{lem:NDbelief} and the single-sender equilibrium condition. By \autoref{lem:rotation} (part \ref{rotation2}), an increase in $\hat s_j$ or a decrease in $p_j$ raises $\abs{U(s_i,p_i,\hat s_j,p_j)-s_i}=U(s_i,p_i,\hat s_j,p_j)-s_i$, i.e., raises sender $i$'s nondisclosure payoff. Following the same argument as in the proof of \autoref{prop:singleCS}, the largest and smallest best-response thresholds must increase.

\emph{Part \ref{subcomp3}:} The proof for $c>0$ is entirely symmetric to the previous part, once we note that $s_i<\hat s^0_i$ implies $s_i-c<\eta(s_i,p_i,\pi)$ because $\hat s^0_i$ is defined as the smallest single-sender equilibrium threshold.
\end{proof}

\begin{proof}[Proof of \autoref{prop:welfare}]
Consider first $c \le 0$. For each sender $i$, define
\begin{equation*}
w_i^{p_1,p_2,c}(\hat s_j) := \min \{\hat s_i \mid U(\hat s_i, p_i,\hat s_j,p_j) \le \hat s_i - c\}.
\end{equation*}
(The minimum exists by continuity of $U$.) That is, $w_i^{p_1,p_2,c}(\cdot)$ gives the smallest best response: the smallest element of ${\hat s_i}^{BR}(\hat s_j,p_i,p_j)$. Let $w^{p_1,p_2,c}:=(w_1^{p_1,p_2,c},w_2^{p_1,p_2,c})$, where $w_1^{p_1,p_2,c}$ depends on $\hat s_2$ and $w_2^{p_1,p_2,c}$ depends on $\hat s_1$. By \autoref{lem:subcomp}, $w^{p_1,p_2}$ is increasing on $[\underline s,\overline s]^2$. By Tarski's fixed point theorem, the smallest fixed point
\begin{equation*}
\mathbf{s}_*(p_1,p_2,c) := \min \{(\hat s_1,\hat s_2)\mid w^{p_1,p_2,c}(\hat s_1,\hat s_2) \le (\hat s_1,\hat s_2)\}
\end{equation*}
exists. We show that $\mathbf{s}_*(p_1,p_2,c)$ is the smallest equilibrium threshold pair, i.e., the smallest fixed point of the best-response correspondence $({\hat s_i}^{BR},{\hat s_j}^{BR})$. Let $\mathbf{s}$ be any equilibrium. Since each sender's threshold is a best response, we have $w(\mathbf{s}) \le \mathbf{s}$. Monotonicity of $w$ implies $w(\mathbf{s}_* \wedge \mathbf{s}) \le w(\mathbf{s}) \le \mathbf{s}$ and $w(\mathbf{s}_* \wedge \mathbf{s}) \le w(\mathbf{s}_*) \le \mathbf{s_*}$, where $\mathbf{s} \wedge \mathbf{s'} := (\min\{s_1,s'_1\},\min\{s_2,s'_2\})$. Thus $w(\mathbf{s}_* \wedge \mathbf{s}) \le \mathbf{s}_* \wedge \mathbf{s}$. The definition of $\mathbf{s}_*$ as a minimum implies $\mathbf{s}_* \le \mathbf{s}_* \wedge \mathbf{s}$, which implies $\mathbf{s}_* \le \mathbf{s}$.

By \autoref{lem:subcomp}, each $w_i^{p_1,p_2,c}$ is decreasing in $p_1$ and in $p_2$ and increasing in $c$. Standard monotone comparative statics \citep[e.g.,][Theorem 6]{MR90} imply that $\mathbf{s}_*(p
_1,p_2,c)$ is decreasing in $p_1$ and $p_2$ and increasing in $c$. A parallel argument to that above, now using $\bar w_i^{p
_1,p_2,c}(\hat s_j) := \sup\{\hat s_i \mid U(\hat s_i, p_i,\hat s_j,p_j) \ge \hat s_i - c\}$ in place of $w_i^{p
_1,p_2,c}$, shows that the largest equilibrium $\mathbf{s}^*(p_1,p_2,c)$ is also decreasing in each $p_i$ and increasing in $c$.

For $c > 0$, let $y_j := -\hat s_j$ and define $\tilde w^{p_i,p_j,c}$ as a self-map on $[\underline s, \overline s] \times [-\overline s, -\underline s]$ by $$\tilde w^{p_i,p_j,c}(\hat s_i, y_j) := (w_i^{p_i,p_j,c}(-y_j), -\bar w_j^{p_i,p_j,c}(\hat s_i)).$$ By \autoref{lem:subcomp}, $w_i^{p_i,p_j,c}$ and $\bar w_j^{p_i,p_j,c}$ are both decreasing in each of their arguments; combined with the sign changes in the definition of $\tilde w^{p_i,p_j,c}$, this implies $\tilde w^{p_i,p_j,c}$ is monotone increasing. The smallest fixed point of $\tilde w^{p_i,p_j,c}$ minimizes $(\hat s_i, y_j) \equiv (\hat s_i, -\hat s_j)$, i.e., it minimizes $\hat s_i$ while maximizing $\hat s_j$, so it is the $j$-maximal equilibrium. \autoref{lem:subcomp} established that $w_i^{p_i,p_j,c}$ is decreasing in $p_i$, while $\bar w_j^{p_i,p_j,c}$ is increasing in $p_i$; hence both components of $\tilde w^{p_i,p_j,c}$ are decreasing in $p_i$. Monotone comparative statics imply the smallest fixed point decreases in $p_i$. The argument for the $i$-maximal equilibrium is symmetric, interchanging the roles of $i$ and $j$. 

Finally, we note that the effect of lower disclosure cost $c>0$ on equilibrium thresholds is ambiguous: it lowers both $w_i^{p_i,p_j,c}$ and $\bar{w}_j^{p_i,p_j,c}$, so the first component of $\tilde{w}^{p_i,p_j,c}$ decreases while the second increases. \autoref{eg:welfare-higher-with-higher-cost} in the \hyperref[app:welfare-example]{Supplementary Appendix} shows that either effect can dominate, and that the implications for DM welfare can also go either way.
\end{proof}

\begin{remark}\label{rem:precision-deriv}Here is the derivation of  \autoref{e:precision-derivative}. Differentiating \eqref{e:NDbelief} with respect to $\rho$ yields (adopting the notational adjustments mentioned before \autoref{e:precision-derivative}):
\begin{align}
\frac{\partial}{\partial \rho} \eta(\hat s,\rho) &= \frac{p}{1 - p + p F^\rho(\hat{s})} \left[ \int_{\underline{s}}^{\hat{s}} s \frac{\partial}{\partial \rho}  f^\rho(s)\mathrm d s - \eta(\hat{s},\rho) \frac{\partial}{\partial \rho} F^\rho(\hat{s}) \right] \notag \\
&= \frac{p}{1 - p + p F^\rho(\hat{s})} \int_{\underline{s}}^{\hat{s}} (s - \eta(\hat s,\rho)) \frac{\partial}{\partial \rho} f^\rho(s) \mathrm d s  \label{e:precision-f-form}\\
&= \frac{p}{1 - p + p F^\rho(\hat{s})} \left[ \left(\hat{s} - \eta(\hat{s}, \rho)\right) \frac{\partial}{\partial \rho} F^\rho(\hat{s}) - \int_{\underline{s}}^{\hat{s}} \frac{\partial}{\partial \rho} F^\rho(s)\mathrm d s \right],\notag
\end{align}
where the second equality uses $\frac{\partial}{\partial \rho} F^\rho(\hat{s}) = \int_{\underline{s}}^{\hat{s}} \frac{\partial}{\partial \rho} f^\rho(s) \mathrm d s$, and the third uses integration by parts and $F^\rho(\underline s)=0$.  
\end{remark}

\begin{remark}
\label{rem:DLR}
Let us show formally that the decreasing likelihood ratio property ensures that \eqref{e:precision-derivative} is negative when $\hat s \equiv \hat s_i \in (\underline s,\pi)$.
Let \mbox{$L^\rho(\hat s):=\E_{F^\rho}[s \mid s<\hat s]$}
and observe that 
\begin{align}
\sign\left[\frac{\partial}{\partial \rho} \eta(\hat s,\rho)\right]&=
\sign\left[\int_{\underline s}^{\hat s}(s-\eta(\hat s,\rho)){\frac{\partial}{\partial \rho}f^\rho(s)} \mathrm d s\right] \notag \\
& =
\sign\left[\int_{\underline s}^{\hat s}\left(s-L^\rho(\hat s)\right){\frac{\partial}{\partial \rho}f^\rho(s)} \mathrm d s+\left(L^\rho(\hat s)-\eta\right){\frac{\partial}{\partial \rho}F^\rho(\hat s)}\right] \notag \\
& = 
\sign\left[F^\rho(\hat s)\frac{\partial}{\partial \rho} L^\rho(\hat s)+\left(L^\rho(\hat s)-\eta(\hat s,\rho)\right){\frac{\partial}{\partial \rho}F^\rho(\hat s)}\right], \label{e:precision-deriv-2}
\end{align}
where the first equality is from \eqref{e:precision-f-form} in \autoref{rem:precision-deriv}, the second equality adds and subtracts $L^\rho(\hat s)$ and uses $\frac{\partial}{\partial \rho} F^\rho(\hat{s}) = \int_{\underline{s}}^{\hat{s}} \frac{\partial}{\partial \rho} f^\rho(s) \mathrm d s$, and the third equality follows from computing $\frac{\partial}{\partial \rho} L^\rho(\hat s)$.
As $\hat s<\pi$, we have ${\frac{\partial}{\partial \rho}F^\rho(\hat s)}>0$ (by the rotation property) and $L^\rho (\hat s)<\pi$ (by the definition of $L^\rho$). The latter inequality implies $L^\rho(\hat s)<\eta(\hat s,\rho)$, since the nondisclosure belief is a weighted average of $L^\rho(\hat s)$  and $\pi$. Hence, $\left(L^\rho(\hat s)-\eta(\hat s,\rho)\right){\frac{\partial}{\partial \rho}F^\rho(\hat s)}<0$.
So 
\eqref{e:precision-deriv-2} is negative if $\frac{\partial}{\partial \rho} L^\rho(\hat s)<0$. This inequality is assured by the decreasing likelihood ratio property, because that implies the conditional distribution of signals on $[\underline s,\hat s]$ is stochastically dominated (in the first-order sense) when $\rho$ increases. Indeed, it is sufficient for higher $\rho$ to yield reverse hazard rate domination on $[\underline s,\pi]$, as that is equivalent to conditional stochastic dominance for all thresholds in the relevant range \citep[(1.B.43) on page 37]{shaked2007stochastic}.
\end{remark}

\subsection{Welfare Examples}
\label{app:welfare-example}

This appendix section substantiates \autoref{cor:main} (part \ref{main3}) and \autoref{prop:welfare} (part \ref{welfare2}) with two examples when there is a disclosure cost. \autoref{eg:welfare-reversal} demonstrates that the DM can be strictly worse off facing two senders than facing either sender alone, while \autoref{eg:welfare-higher-with-higher-cost} demonstrates that when facing two senders, the DM's welfare can be nonmonotonic in the disclosure cost. Details verifying the claims for these examples are in the \hyperref[app:welfare-example-proofs]{Supplementary Appendix}.

\begin{example}
\label{eg:welfare-reversal}
The prior is $\pi=1/2$. The information structure is parametrized by $\gamma \in (1/2,1)$ and $\delta\in (0,1)$. There are four signals,\footnote{Although this violates our maintained assumption of continuous signals, one can perturb the example to make it continuous while preserving the conclusion.} with conditional probabilities in each state given by:
\begin{center}
    \begin{tabular}{c|cccc}
     & $\underline{s}$ & $s^l$ & $s^h$ & $\bar{s}$ \\ \hline
    $\omega=0$ & $1-\delta$ & $\gamma \delta$ & $(1-\gamma)\delta$ & $0$ \\
    $\omega=1$ & $0$ & $(1-\gamma)\delta$ & $\gamma \delta$ & $1-\delta$
    \end{tabular}.
\end{center}
There is a cost of disclosure, $c>0$. The DM chooses an action $a\in [0,1]$ to maximize $-(a-\omega)^2$.

As we equate signals with their posteriors on state $1$, $$\underline{s}=0 < s^l=1-\gamma < 1/2 < s^h=\gamma<\bar{s}=1.$$ 
Because of the example's discrete signals, a threshold sender strategy may now entail randomization at one signal; with that caveat, equilibria must still be in threshold strategies. Observe that---regardless of whether there are one or two senders, and regardless of whether they are upward or downward biased---any sender must withhold the two signals that are least favorable to him (i.e., $\underline s$ and $s^l$ if upward biased, and $s^h$ and $\overline s$ if downward biased) in any equilibrium; this is because $s^l<\pi<s^h$. Also observe that by symmetry of the problem, the DM's payoff in the best equilibrium with a single sender does not depend on the direction of that sender's bias.

We now specialize to the parameters $\gamma = 0.7$, $\delta = 0.7$, $p_1 = p_2 = 0.8 \equiv p$, and $c = 0.36$. 

\begin{claim}
\label{claim:single-sender-benchmark}
With a single upward-biased sender, the DM's best equilibrium has the sender disclosing $s^h$ and $\bar{s}$ and concealing $\underline s$ and $s^l$. The DM's expected payoff in this equilibrium is $-0.1864$.
\end{claim}

\begin{claim}
\label{claim:two-senders-worse}
With two opposite-biased senders, the unique equilibrium has each sender disclosing his most favorable signal and concealing the other three signals. The DM's expected payoff in this equilibrium is $-0.19$.
\end{claim}

Thus, the DM's welfare in the best single-sender equilibrium is strictly higher than in the unique equilibrium with two opposite-biased senders.
\end{example}

\begin{example}
\label{eg:welfare-higher-with-higher-cost}
Now modify \autoref{eg:welfare-reversal} to a higher cost $c=0.38$, keeping all other parameters the same.

\begin{claim}
\label{claim:higher-cost-better}
When $c=0.38$ (all other parameters unchanged from \autoref{eg:welfare-reversal}), there is an equilibrium in which the DM's welfare equals that of her best equilibrium from the single-sender benchmark. Thus, a higher disclosure cost can improve DM welfare in the best equilibrium with two senders.
\end{claim}

Plainly, if the cost $c$ is raised sufficiently high, then the only equilibrium involves no disclosure from either agent, which would yield the DM a lower welfare than in the two-sender equilibrium from \autoref{eg:welfare-reversal}. The upshot is that the DM's welfare (in her best equilibrium) with two senders can be nonmonotonic in the disclosure cost.
\end{example}

Finally, although we omit the verification, we note that \autoref{eg:welfare-reversal}'s parameters can be modified to demonstrate that the DM can be strictly worse off facing two similarly-biased senders than facing either sender alone. That can be verified with parameters $\gamma = 0.8$, $\delta = 0.7$, $p_1 = p_2 = 0.4$, $c = 0.385$.
The takeaway is that the message of \autoref{eg:welfare-reversal} does not depend on senders being biased in opposite directions.

\subsection{Perfectly correlated signals}
\label{app:correlated}

The assumption that senders have conditionally independent signals is important for our analysis, as it permits us to apply the IVP result from \citet*{KLS21-IVP}.

Consider now the polar opposite case in which informed senders' signals are perfectly correlated and, for simplicity, there is no message cost ($c=0$). Specifically, there is a single signal $s$ drawn from distribution $F(s|\omega)$, and each sender $i$ is independently either informed of $s$ with probability $p_i$ or remains uninformed. This setting is effectively identical to the ``extreme agenda'' case of \citet{BM13}.\footnote{They allow senders' utility functions to be non-linear; the following discussion does not depend on linearity. Note that \citet{BM13} assume, as we do, that whether a sender is informed is independent of the other sender; see \citet{BGM18} for correlated information endowments.}

If both senders are biased in the same direction, the model can be mapped to a single-sender problem where the sender is informed with probability $p_1+p_2-p_1p_2 > \max\{p_1,p_2\}$.\footnote{Perfect correlation implies there is only one relevant nondisclosure belief, viz., when both senders don't disclose. Senders with the same bias must use the same equilibrium threshold. Given any such threshold, the nondisclosure belief is computed as in \eqref{e:NDbelief} (assuming upward bias), but with $1-p$ replaced by the probability that both senders are uninformed, i.e., $(1-p_1)(1-p_2)$.} \autoref{prop:singleCS} then implies that each sender discloses more when there is an additional sender; hence, the DM is always better off with two senders.

It is instructive to understand why our irrelevance result (\autoref{lem:subcomp}, part \ref{subcomp1}) no longer holds. Suppose both senders are upward biased and symmetric ($p_1=p_2\equiv p$). Let $\hat s^0$ denote the common single-sender threshold, so that $\eta(\hat s^0,p,\pi)=\hat s^0$. When sender $j$ is added with the hypothesis that he too discloses all signals weakly above $\hat s^0$, type $\hat s^0$ of sender $i$ no longer expects the DM's belief to be $\hat s^0$ should he conceal. Rather, he expects the DM's belief to be strictly lower: if $j$ is informed, the DM's belief will be $\hat s^0$, but if $j$ is uninformed, the DM's belief will be strictly lower because of nondisclosure from two senders rather than one. This makes type $\hat s^0$ strictly prefer disclosure. From the perspective of applying \citet*{KLS21-IVP}'s IVP result as in \autoref{lem:rotation}, the point is that under correlated signals, when an informed sender $i$ conceals, sender $i$ and the DM do not agree on the experiment generated by $j$'s message. Thus, even if the DM's nondisclosure belief equals $i$'s belief (over the state), $i$'s expectation of the DM's posterior can differ.

Interestingly, welfare conclusions under perfectly correlated signals are very different when senders have opposing biases. The following proposition shows that each sender discloses strictly less than in his single-sender game.

\begin{proposition}
\label{prop:correlated-opposite-worse}
Assume perfectly correlated signals, no message cost, and that the two senders have opposing biases. Each sender then discloses strictly less than in his single-sender game.
\end{proposition}

\begin{proof}
We prove it for the upward-biased sender; the argument is symmetric for the other sender. Let sender 1 be upward biased and sender 2 be downward biased. Let $\hat s^0_1$ denote the single-sender threshold, i.e., $\eta(\hat s^0_1)=\hat s^0_1$, where we have suppressed the dependence of $\eta$ on $p_1$. Write $\eta(\hat s_1,\hat s_2)$ as the nondisclosure belief (in the event neither sender discloses) in the two-sender game when the respective thresholds are $\hat s_1$ and $\hat s_2$. Although the DM's updating is not separable, it is clear that $\eta(\hat s_1,\hat s_2)\geq \eta(\hat s_1)$, with equality if and only if $\hat s_2=\underline s$. This follows because the nondisclosure event is the union of: (i) $m_1=m_2=\phi$ and $s_2=\phi$; and (ii) $m_1=m_2=\phi$ and $s_2>\hat s_2$. Conditional on the first event, the DM's posterior is $\eta(\hat s_1)$; conditional on the second, the posterior is larger (strictly if $\hat s_2 > \underline s$). Since $\eta(\hat s_1)\geq \hat s_1$ for all $\hat s_1\leq s^0_1$, it follows that for any $\hat s_2>\underline s$, if the DM conjectures thresholds $(\hat s_1,\hat s_2)$, then sender 1 with signal $s_1=\hat s_1$ strictly prefers nondisclosure to disclosure. Therefore, since sender 2 uses a threshold strictly larger than $\underline s$ in any equilibrium, any equilibrium involves sender 1's threshold being strictly larger than $\hat s^0_1$.
\end{proof}

Thus, despite no message costs and the increased availability of information, overall disclosure is not more informative in the Blackwell sense than under either single-sender problem. Consequently, there exist DM preferences such that she strictly prefers to face either sender alone rather than both simultaneously. This implies that the welfare conclusion in Corollary 2 of \citet{BM13} can be reversed under alternative DM preferences.

In general, for an arbitrary message cost $c$, an interior equilibrium $(\hat s_1,\hat s_2)$ requires
$$\Pr[m_2\neq \phi \mid s_1=\hat s_1,\hat s_2]\hat s_1+\Pr[m_2= \phi \mid s_1=\hat s_1,\hat s_2]\eta(\hat s_1,\hat s_2)= \hat s_1-c,$$
or equivalently,
$$\left(\hat s_1-\eta(\hat s_1,\hat s_2)\right)=\frac{c}{\Pr\big[m_2=\phi \mid s_1=\hat s_1,\hat s_2\big]}.$$
Thus, when $c \geq  0$, using $\eta(\hat s_1,\hat s_2)> \eta(\hat s_1)$ (strict by interiority of $\hat s_2$), it follows that $\hat s_1 > \eta(\hat s_1)$. Furthermore, for any $c \geq 0$, there is an equilibrium in which $\hat s_1$ is weakly larger than the largest single-sender equilibrium; this is consistent with \citepos{EF19} finding that with a disclosure cost ($c>0$), the DM might be better off by barring one sender. With a concealment cost $(c<0)$, the comparison between equilibrium thresholds with two senders versus one sender is ambiguous.

\subsection{Non-linear utility functions}
\label{app:nonlinear}

Another important assumption to straightforwardly apply the IVP result from \citet*{KLS21-IVP} is that each sender has linear preferences. Suppose, more generally, that sender $i$'s utility is given by some function $V_i(\beta_{DM})$. The comparative statics of sender $i$'s disclosure depend on how
\begin{align}
\label{eq:nonlinear}
\E\big[V_i(T(\beta_i,\mu_{DM},\mu_i))\big]-\E\big[V_i(\beta_i)\big]
\end{align}
varies across Blackwell-comparable experiments, where the expectation is over the posterior $\beta_i$ using $i$'s beliefs and $T(\cdot)$ is the transformation in \eqref{e:transform}. When $V_i(\cdot)$ is linear, the second term in \eqref{eq:nonlinear} is constant across experiments, and the IVP result tells us that the sign of the change in the first term is determined by the sign of $\mu_i-\mu_{DM}$, i.e., the disagreement in the interim beliefs. Unambiguous comparative statics of \eqref{eq:nonlinear} cannot be obtained for arbitrary $V_i(\cdot)$. However, because $T(\beta_i,\mu,\mu)=\beta_i$ for any $\beta_i$ and $\mu$, the logic behind our irrelevance result extends quite generally:

\begin{proposition}
\label{prop:generalV}
If $c=0$ and $V_i(\cdot)$ is strictly monotone, then no matter $j$'s disclosure strategy, the best-response threshold for $i$ is the same as when he is a single sender.
\end{proposition}

\begin{proof}
Denote $r:= \frac{1-\mu_{DM}}{\mu_{DM}}\frac{\mu_i}{1-\mu_{i}}$, and define $W(\beta,r):=V_i(T(\beta,r))-V_i(\beta)$, where $T(\beta,r):= \frac{\beta}{\beta+(1-\beta)r}$ is shorthand for $T(\beta,\mu_{DM},\mu_{i})$. When $V_i(\cdot)$ is strictly monotone and $c=0$, $i$'s best-response threshold must satisfy $\E [W(\cdot,r)]=0$ when $r$ is determined by $i$'s threshold type and the DM's nondisclosure belief.

When $r=1$, we have for any $\beta$ that $T(\beta,1)=\beta$ and hence $W(\beta,1)=0$. Because $T(\beta,r)$ is strictly decreasing in $r$ for all interior $\beta$, any experiment that is not fully informative has $\E [W(\cdot,r)]=0$ if and only if $r=1$. Thus, no matter $j$'s disclosure strategy (so long as it is not fully informative, which it cannot be since $p_j<1$), $i$'s best-response threshold is such that $r=1$, i.e., the DM's nondisclosure belief equals $i$'s threshold type. But this is the same condition as in the single-sender game.
\end{proof}

When $c\neq 0$, there are non-linear specifications for $V_i(\cdot)$ under which our themes about strategic complementarity under concealment cost or substitutability under disclosure cost do extend, and there are other specifications which make conclusions ambiguous or even reversed. We illustrate in the \hyperref[app:nonlinear-supp]{Supplementary Appendix} through a family of power utility functions.

\bibliographystyle{ecta}
\bibliography{disclosure}

\newpage

\begin{center}
\textbf{{\LARGE {Supplementary Appendix}}}
\end{center}

\section{Non-linear Utilities}
\label{app:nonlinear-supp}

This appendix expands on the discussion in \appendixref{app:nonlinear}. Specifically, we show through a family of power utility functions how our main conclusions can be affected by departures from linearity of a sender's utility $V_i(\beta_{DM})$.

To that end, define
\begin{equation}
\label{e:W}
W(\beta,r):=V_i(T(\beta,r))-V_i(\beta),
\end{equation}
as in the proof of \autoref{prop:generalV}; here, $r\equiv \frac{1-\mu_{DM}}{\mu_{DM}}\frac{\mu_i}{1-\mu_{i}}$ is a measure of the disagreement in interim beliefs $\mu_{DM}$ and $\mu_i$, and $T(\beta,r)\equiv \frac{\beta}{\beta+(1-\beta)r}$ is shorthand for the transformation $T(\beta,\mu_{DM},\mu_{i})$ from \eqref{e:transform}. Under a disclosure cost $(c>0)$ the relevant case is $r>1$ if the sender is upward biased (as in this case the threshold sender type must have a higher belief than the DM's nondisclosure belief) and $r<1$ if the sender is downward biased; under a concealment cost $(c<0)$ the relevant case is $r<1$ if the sender is upward biased and $r>1$ if the sender is downward biased.

In the following proposition, we say that an upward-biased sender $i$'s disclosure is a strategic substitute (resp., complement) to $j$'s if when $j$'s message is more Blackwell-informative, then $i$'s largest and smallest best response disclosure thresholds increase (resp., decrease).

\begin{proposition}
\label{prop:exponential}
Assume $V_i(\beta)=\gamma \beta^\alpha$, where either $\gamma,\alpha>0$ or $\gamma,\alpha<0$, so that sender $i$ is upward biased. Then $i$'s disclosure is:
\begin{enumerate}
\item \label{exponential1} a strategic substitute to $j$'s under disclosure cost if $0<\alpha\leq 1$ and $\gamma>0$;
\item \label{exponential4} a strategic complement to $j$'s under a small enough concealment cost if $0 < \alpha \leq 1$ and $\gamma > 0$;\footnote{A small enough concealment cost means that $c<0$ is close enough to $0$ that the interim disagreement condition in part \ref{exponential-curvature-4} of \autoref{lem:exponential-curvature} holds. Note that as $c \to 0$, all of $i$'s best response thresholds converge uniformly to the unique single-sender equilibrium threshold (at which point $r=1$), so $r$ is assured to be in the relevant range when $|c|$ is small enough.}
\item \label{exponential2} a strategic complement to $j$'s under disclosure cost if $\alpha< -1$ and $\gamma<0$;
\item \label{exponential3} a strategic substitute to $j$'s under concealment cost if $\alpha < 0$ and $\gamma < 0$.
\end{enumerate}
\end{proposition}

Parts \ref{exponential1} and \ref{exponential4} of \autoref{prop:exponential} extend, respectively, parts \ref{subcomp2} and \ref{subcomp3} of \autoref{lem:subcomp} to some non-linear preferences (and subject to a small enough concealment cost, in part \ref{exponential4}). Parts \ref{exponential2} and \ref{exponential3} of \autoref{prop:exponential} show how our findings of strategic complementarity under concealment cost and strategic substitutability under disclosure cost can be reversed for other non-linear preferences. Note that one must be careful with the analog for a downward-biased sender, because the direction of disagreement reverses. If $V_i(\beta)=-\gamma \beta^\alpha$, then in each part above, ``disclosure cost'' and ``concealment cost'' should be interchanged.

Given the discussion in the first paragraph of \appendixref{app:nonlinear} and following the logic given in the main text after \autoref{lem:rotation}, \autoref{prop:exponential} is a straightforward consequence of the following lemma. Recall that a disclosure cost corresponds to $r>1$ and a concealment cost to $r<1$.

\begin{lemma}
\label{lem:exponential-curvature}
If $V_i(\beta)=\beta^\alpha$, then $W(\beta,r)$ defined in \eqref{e:W} is:
\begin{enumerate}
\item convex in $\beta$ if $0<\alpha\leq 1$ and $r>1$;
\item \label{exponential-curvature-4} concave in $\beta$ if $0 < \alpha \leq 1$ and $\frac{1-\alpha}{1+\alpha} < r < 1$;
\item convex in $\beta$ if $\alpha < -1$ and $r>1$;
\item \label{exponential-curvature-3} concave in $\beta$ if $\alpha < 0$ and $r<1$.
\end{enumerate}
\end{lemma}

\begin{proof}[Proof of \autoref{lem:exponential-curvature}]
Denoting partial derivatives with subscripts as usual, we compute that $
W_{\beta\beta}(\cdot)$ is equal to
\begin{equation*}
V''_i\left(\frac{\beta}{\beta+(1-\beta)r}\right)\left(\frac{r^2}{(\beta+(1-\beta)r)^4}\right)+V'_i\left(\frac{\beta}{\beta+(1-\beta)r}\right)\left(\frac{2r(r-1)}{(\beta+(1-\beta)r)^3}\right)-V''_i(\beta).
\end{equation*}
Plugging in $V_i(\beta)=\beta^\alpha$ and doing some algebra shows that $
W_{\beta\beta}(\cdot)$ has the same sign as:
\begin{equation*}
\alpha \left[ (1-\alpha )+\frac{r (2\beta(r-1)-r (1-\alpha))}{\left(\beta+(1-\beta)r\right)^{\alpha +2}} \right]=:H(\beta,\alpha,r).
\end{equation*}
Observe that $H(0,\alpha,r)=\alpha(1-\alpha)(1-r^{-\alpha})$, and hence if $\alpha < 1$ and $\alpha\neq 0$ then $\sign[H(0,\alpha,r)]=\sign[r-1]$.  Differentiating yields
\begin{equation*}
H_\beta(\cdot)=\frac{\alpha (\alpha +1) (r-1) r (\alpha  r+2 \beta  (r-1))}{(\beta+(1 -\beta)r)^{\alpha +3}}.
\end{equation*}
We now consider five cases; they correspond to the lemma's enumeration, except that the final two cases here together establish part \ref{exponential-curvature-3} of the lemma.
\begin{enumerate}[leftmargin=*]
\item Suppose $0<\alpha \leq 1$ and $r > 1$.  Then $H(0,\alpha,r)\geq 0$ and $H_\beta(\cdot)>0$, and hence $H(\beta,\alpha,r)>0$ for all $\beta\in (0,1)$.
\item Suppose $0 < \alpha \leq 1$ and $\frac{1-\alpha}{1+\alpha} < r < 1$.  Then $H(0,\alpha,r)=\alpha(1-\alpha)(1-r^{-\alpha})\leq 0$.  
The roots of $H(1,\alpha,r)=\alpha\left[(1-\alpha)+r^2(1+\alpha)-2 r\right]$ viewed as a function of $r$ are $r = 1$ and $r = \frac{1-\alpha}{1+\alpha}$; since $H(1,\alpha,r)$ is convex in $r$, it is negative between the roots, and hence $H(1,\alpha,r) < 0$ for $r$ in the specified range.  Furthermore, $\sign[H_\beta(\cdot)] = \sign[2\beta(1-r)-\alpha r]$; hence, since the expression $2\beta(1-r)-\alpha r$ is negative at $\beta = 0$ and increasing in $\beta$, any interior critical point of $H(\cdot,\alpha,r)$ is a minimum.  It follows that $H(\beta,\alpha,r) < 0$ for all $\beta \in (0,1)$.
\item \label{exponential-curvature-case-3} Suppose $\alpha < -1$ and $r>1$.  Then $H(0,\alpha,r)> 0$ and \[H(1,\alpha,r)=\alpha (r-1)(\alpha-1+r(1+\alpha))>0.\]  We will show that $H_\beta(\beta,\alpha,r)=0$ implies $H(\beta,\alpha,r)>0$, which combines with the previous two inequalities to imply that $H(\cdot)>0$.  Accordingly, assume $H_\beta(\beta,\alpha,r)=0$, which occurs when $\beta=\frac{\alpha r}{2(1-r)}$, which implies $\alpha \in(-2,-1)$ and $r\geq \frac{2}{2+\alpha}$ (because $\beta\leq 1$ and $\alpha<-1$) .  Furthermore,
\begin{equation*}
H\left(\frac{\alpha r}{2(1-r)},\alpha,r\right)
=\alpha \left[1-\alpha-r^2\left(\frac{2}{r(\alpha+2)}\right)^{\alpha+2}\right].
\end{equation*}
The derivative of the above expression with respect to $r$ is $\alpha^2 \left(\frac{2}{\alpha+2}\right)^{\alpha+2} r^{-\alpha-1}$, which is strictly positive given $\alpha \in (-2,-1)$ and $r>\frac{2}{\alpha+2}>2$.  Moreover, when evaluated with $r=2$, the expression reduces to $\alpha\left[1-\alpha-\frac{4}{(\alpha+2)^{\alpha+2}}\right]$.  For $\alpha \in (-2,-1)$, the bracketed term is negative because $(\alpha+2)^{\alpha+2}<1$ implies $\frac{4}{(\alpha+2)^{\alpha+2}}>4$, while $1-\alpha<3$. Since $\alpha<0$, we have $H\left(\frac{\alpha r}{2(1-r)},\alpha,r\right)>0$, as was to be shown.
\item Suppose $-1\leq \alpha <0$ and $0\leq r<1$.  Then $H(0,\alpha,r)< 0$ and $H_\beta(\cdot)\leq 0$, and hence $H(\beta,\alpha,r)<0$ for all $\beta\in (0,1)$.
\item Suppose $\alpha < -1$ and $0\leq r<1$.  Then $H(0,\alpha,r)<0$ and \[H(1,\alpha,r)=\mbox{$\alpha(r-1)(\alpha-1+r(1+\alpha))$}<0.\]  As argued in case \ref{exponential-curvature-case-3} above, $H_\beta(\beta,\alpha,r)=0$ requires $\alpha \in(-2,-1)$ and $r\geq \frac{2}{2+\alpha}>2$, which is not possible given that we have assumed $r<1$.  Thus, on the relevant domain, $H_\beta(\cdot)$ has a constant sign and so $H(\cdot)<0$. \qedhere
\end{enumerate}
\end{proof}

\section{Details for Welfare Examples}
\label{app:welfare-example-proofs}

Here we provide the verifications for \autoref{eg:welfare-reversal} and \autoref{eg:welfare-higher-with-higher-cost}. Given the examples' discrete signals, we will abuse notation and say that a sender uses threshold $1/2$ if he reveals his two favorable signals and conceals the two unfavorable ones. Note that in \autoref{eg:welfare-reversal} with $c=0.36$, if there is just one upward-biased sender, then disclosing $s^h$ yields that sender a payoff $\gamma - c = 0.34$ and disclosing $\bar{s}$ yields $1 - c = 0.64$.

\begin{proof}[Proof of \autoref{claim:single-sender-benchmark}]
To establish the best equilibrium, it suffices to show that there is an equilibrium in which the sender reveals $s^h$ and $\bar s$ (i.e., uses threshold $1/2$). The corresponding nondisclosure belief is
\begin{equation}
\label{e:ND-maximal-reveal}
\eta(1/2, p, \pi) = \frac{1-p + p(1-\gamma)\delta}{2-p} = 0.3067.
\end{equation}
Since $0.3067 < \gamma - c = 0.34$, the sender strictly prefers to disclose $s^h$, which verifies that we have an equilibrium.

To calculate welfare, observe that because of her quadratic-loss objective, the DM's welfare in any equilibrium is $-\mathbb{E}[\mu(1-\mu)]$, where $\mu$ is the posterior. In the above best equilibrium, disclosure of $\bar{s}$ is perfectly informative, disclosure of $s^h$ yields posterior $\gamma$, and nondisclosure yields the posterior in \eqref{e:ND-maximal-reveal}, which we denote below as $\eta^*$. The DM's welfare is thus
\begin{equation*}
\label{e:welfare-1}
-\left(1-\frac{p}{2}\right)\eta^*(1-\eta^*) - \frac{p\delta}{2}\gamma(1-\gamma) = -0.1864. \qedhere
\end{equation*}
\end{proof}

\begin{proof}[Proof of \autoref{claim:two-senders-worse}]
Let sender 1 be upward biased and sender 2 be downward biased. We proceed in three steps.

\noindent \emph{Step 1: Each sender discloses his most favorable signal.}
We argue it for sender 1; it is symmetric for sender 2. Given that sender 2 must conceal his two unfavorable signals ($s^h$ and $\bar{s}$), the hardest scenario for sender 1 to want to disclose $\bar{s}$ is when: (i) he is conjectured to conceal all signals (maximizing the nondisclosure belief), and (ii) sender 2 discloses $s^l$ and $\underline{s}$ (by strategic substitution, the more sender 2 discloses, the less sender 1's incentive to disclose). In this scenario, sender 1's concealment is uninformative, so if both senders conceal, the DM's posterior is $\frac{1-p + p(\delta\gamma + 1-\delta)}{2-p}$ (analogously to \eqref{e:ND-maximal-reveal}). If sender 2 discloses $s^l$, the DM's posterior is $1-\gamma$. Thus, sender 1's expected payoff from concealing $\bar{s}$ is
\[
\delta(1-\gamma)p(1-\gamma) + \bigl(1-\delta(1-\gamma)p\bigr)\frac{1-p + p(\delta\gamma + 1-\delta)}{2-p} = 0.6273.
\]
Since $0.6273 < 1 - c = 0.64$, sender 1 strictly prefers to disclose $\bar{s}$.

\noindent \emph{Step 2: Each sender conceals all other signals.}
Again, we only argue it for sender 1. It suffices to argue that sender 1 conceals $s^h$. Given Step 1, the easiest scenario for sender 1 to want to disclose $s^h$ is: (i) he is conjectured to disclose $s^h$ and $\bar s$, i.e., use threshold $1/2$ (minimizing the nondisclosure belief), and (ii) sender 2 plays the least revealing strategy consistent with Step 1, i.e., conceals $s^l,s^h$, and $\bar s$. In this scenario, if sender 2 discloses, it must be signal $\underline{s}$, giving sender 1 a payoff of $0$. If both senders conceal, the DM's posterior is
\[
\frac{\eta(1/2, p, \pi)}{\eta(1/2, p, \pi) + \bigl(1-\eta(1/2, p, \pi)\bigr)(1-p + p\delta)},
\]
where $\eta(1/2, p, \pi)$ is given in \eqref{e:ND-maximal-reveal}.
Thus, sender 1's expected payoff from concealing $s^h$ is
\[
\bigl(1- (1-\gamma)(1-\delta)p\bigr)\frac{\eta(1/2, p, \pi)}{\eta(1/2, p, \pi) + \bigl(1-\eta(1/2, p, \pi)\bigr)(1-p + p\delta)}
 = 0.3414.
\]
Since $0.3414 > \gamma - c = 0.34$, sender 1 strictly prefers to conceal $s^h$ even in the scenario most conducive to disclosing it.

\noindent \emph{Step 3: Welfare.} Steps 1 and 2 (combined with equilibrium existence) establish that there is a unique equilibrium: each sender reveals their most favorable signal and conceals the other three.\footnote{Indeed, to confirm this is an equilibrium, fix those sender strategies. If sender 2 discloses, it must be signal $\underline{s}$, giving sender 1 a payoff of $0$. If both senders conceal, symmetry implies the DM's posterior is $1/2$. Sender 1's expected payoff from concealing $s^h$ is thus
\[
\bigl(1-p(1-\gamma)(1-\delta)\bigr)(1/2) = 0.464.
\]
Since $0.464 > \gamma - c = 0.34$, sender 1 prefers to conceal $s^h$. So sender 1 is playing a best response. By a symmetry argument, so is sender 2.} To calculate the equilibrium welfare, we follow the approach used to calculate welfare in the proof of \autoref{claim:single-sender-benchmark}. In the current two-sender equilibrium, the only disclosures are extreme signals (posteriors $0$ or $1$), which contribute zero variance. Nondisclosure yields posterior $1/2$ by symmetry. The DM's welfare with the two senders is thus
\[
-\bigl(1-p(1-\delta)\bigr)\tfrac{1}{4} = -0.19. \qedhere
\]
\end{proof}

\begin{proof}[Proof of \autoref{claim:higher-cost-better}]
It suffices to verify that there is an equilibrium in which the upward-biased sender 1 reveals signals $s^h$ and $\bar s$ (and conceals the other two signals) while the downward-biased sender 2 conceals all signals, as the DM's welfare in this equilibrium is the same as in the single-sender analysis of \autoref{eg:welfare-reversal}.

So consider those strategies and corresponding DM beliefs. Then sender 2's payoff from concealing $\underline{s}$ is $1-0.3727=0.6273$ (because, by symmetry to the calculation in Step 1 of the proof of \autoref{claim:two-senders-worse}, the DM's expected posterior is $1-0.06273=0.3727$). Since $0.6273 > 0.62=1- c$, sender 2 strictly prefers to conceal his most favorable signal, and concealing all signals is his best response. Turning to sender 1: he effectively faces a single-sender problem with nondisclosure belief $0.3067$ from \eqref{e:ND-maximal-reveal}. Since $0.3067 < \gamma - c = 0.32$, sender 1 strictly prefers to disclose $s^h$, and hence he is also playing his best response.
\end{proof}
\end{document}